\documentclass[]{lmcs}
\pdfoutput=1
\usepackage[utf8]{inputenc}

\usepackage{lastpage}
\lmcsdoi{19}{2}{10}
\lmcsheading{}{\pageref{LastPage}}{}{}%
{Mar.~15,~2022}{May~11,~2023}{}

\usepackage{amssymb}
\usepackage{centernot}
\usepackage{forest}
\usepackage{tikz}
\usepackage{xcolor}
\usepackage{mathrsfs}  

\usetikzlibrary{arrows, automata, positioning, decorations.pathreplacing}

\tikzset{
  every picture/.style={
    ->, >=stealth', shorten >=0.5pt,  auto, inner sep=2pt,
    semithick, double distance=1.25pt},
  every state/.style={
    circle, minimum size=18pt, inner sep=1pt, draw},
  every loop/.style={
    looseness=8, min distance=12.5mm},
} 

\DeclareFontFamily{U}{mathb}{\hyphenchar\font45}
\DeclareFontShape{U}{mathb}{m}{n}{
      <5> <6> <7> <8> <9> <10> gen * mathb
      <10.95> mathb10 <12> <14.4> <17.28> <20.74> <24.88> mathb12
      }{}
\DeclareSymbolFont{mathb}{U}{mathb}{m}{n}
\DeclareMathSymbol{\preccurlyeq}{3}{mathb}{"A4}
\DeclareMathSymbol{\shortminus}{\mathbin}{AMSa}{"39}

\newcommand{\blankline}{\vspace{1.0\baselineskip}}

\newcommand{\Act}{\calA}
\newcommand{\Bit}{\mathbb{B}}
\newcommand{\ES}{\mathit{es}}
\newcommand{\IRC}{\mathit{rc}}
\newcommand{\PIRC}{\mathit{prc}}
\newcommand{\Nat}{{\mathbb{N}}}

\newcommand{\apair}[2]{\langle {#1}, {#2} \rangle}
\newcommand{\bpair}[2]{\lbrack \mkern1mu {#1}, {#2} \mkern1mu \rbrack}
\newcommand{\bigo}[1]{{O\!\left(#1\right)}}
\newcommand{\bigO}{\mathit{O}}   
\newcommand{\bit}[1]{\textup{\texttt{#1}}}
\newcommand{\calA}{\mathscr{A}}
\newcommand{\calB}{\mathscr{B}}
\newcommand{\calC}{\mathscr{C}}
\newcommand{\calD}{\mathscr{D}}
\newcommand{\calG}{\mathscr{G}}
\newcommand{\calN}{\mathscr{N}\!}
\newcommand{\calS}{\mathscr{S}}
\newcommand{\calNplus}{\mathscr{N}^+}
\newcommand{\ofA}{\mathbb{A}}
\newcommand{\ofN}{\mathbb{N}}
\newcommand{\lbl}{\mathit{lbl}}
\newcommand{\lc}{\mathopen{\lbrace \,}}
\newcommand{\rc}{\mathclose{\, \rbrace}}
\newcommand{\piES}{\pi_{\mkern-1mu \ES} \mkern1mu}

\newcommand{\pijl}[1]{\mathrel{\text{$\xrightarrow{\smash[t]{#1}}$}}}
\newcommand{\bpijl}[1]{\mathrel{\text{$\xrightarrow{\smash[t]{#1}}_{\calB}$}}}
\newcommand{\cpijl}[1]{\mathrel{\text{$\xrightarrow{\smash[t]{#1}}_{\calC}$}}}

\newcommand{\pijlster}[1]{\mathrel{\text{$\xrightarrow{\smash[t]{#1}}^{\mkern-1mu \ast}$}}}
\newcommand{\pijlsterc}[1]{\mathrel{\text{$\xrightarrow{\smash[t]{#1}}^{\mkern-1mu \ast}_{\calC}$}}}

\newlength{\leftrightarrowwidth}
\newcommand{\bisauxiliary}{%
  \raisebox{.3ex}{%
    \makebox[\leftrightarrowwidth]{%
      $\underline{\makebox[0.7\leftrightarrowwidth]{$\leftrightarrow$}}$
    }}}  
\newcommand{\superbis}{\mbox{\raisebox{2pt}{\scalebox{0.7}{$\bisauxiliary$}}}}
\newcommand{\bis}{\mathrel{\bisauxiliary}}

\title{Lowerbounds for Bisimulation by Partition Refinement} 

\author[J.F.~Groote]{Jan Friso Groote\lmcsorcid{0000-0003-2196-6587}}
\address{Eindhoven University of Technology, The Netherlands}
\email{j.f.groote@tue.nl, j.j.m.martens@tue.nl, e.p.d.vink@tue.nl}

\author[J.J.M.~Martens]{Jan Martens\lmcsorcid{0000-0003-4797-7735}}
\thanks{Partially funded by the AVVA project NWO 612.001.751/TOP1.17.002}

\author[E.P.~de~Vink]{Erik P. de Vink\lmcsorcid{0000-0001-9514-2260}}


\begin{document}
\maketitle

\begin{abstract}
  We provide time lowerbounds for sequential and parallel algorithms
  deciding bisimulation on labelled transition systems that use
  partition refinement. For sequential algorithms this is
  $\Omega((m \mkern1mu {+} \mkern1mu n ) \mkern-1mu \log \mkern-1mu
  n)$ and for parallel algorithms this is $\Omega(n)$, where $n$ is
  the number of states and $m$ is the number of transitions. The
  lowerbounds are obtained by analysing families of deterministic
  transition systems, ultimately with two actions in the sequential
  case, and one action for parallel algorithms.

  For deterministic transition systems with one action, bisimilarity
  can be decided sequentially with fundamentally different techniques
  than partition refinement. In particular, Paige, Tarjan, and Bonic
  give a linear algorithm for this specific situation. We show,
  exploiting the concept of an oracle, that this approach is not of
  help to develop a faster generic algorithm for deciding
  bisimilarity. For parallel algorithms there is a similar situation
  where our approach can be applied, too.
\end{abstract}


\section{Introduction}
\label{sec-intro}

Strong bisimulation~\cite{park1981gi,milner1980calculus} is the gold
standard for equivalence on labelled transition systems
(LTSs). Deciding bisimulation equivalence among the states of an LTS
is a crucial step for tool-supported analysis and model checking of
LTSs. The well-known and widely-used partition refinement algorithm of
Paige and Tarjan~\cite{paige1987three} has a worst-case upperbound
$\bigO(m \mkern-1mu \log \mkern-1mu n)$ for establishing the
bisimulation equivalence classes. Here, $n$~is the number of states and $m$~is the number of
transitions in an LTS\@. 

The algorithm of Paige and Tarjan seeks to find, starting from an
initial partition, via refinement steps, the coarsest stable
partition, that in fact is built from the bisimulation equivalence
classes that are looked for. The algorithm achieves the complexity of
the logarithm of the number of states times the number of transitions 
by restricting the amount of work for refining blocks and moving states.
When refining, the splitting blocks are investigated using an
intricate bookkeeping trick. Only the smaller parts of a block that
are to be moved to a new block are split off, leaving the bulk of the
original block at its place. These specific ideas go back
to~\cite{hopcroft1971DFAmin} and make the difference with the earlier
$\bigO(m \mkern0.5mu n)$ algorithm of Kanellakis and
Smolka~\cite{kanellakis1983}. The algorithms by Kanellakis-Smolka and
Paige-Tarjan, with the format of successive refinements of an initial
partition till a fixpoint is reached, have been leading for variations
and generalisations for deciding specific forms of (strong)
bisimilarities, see e.g.\ \cite{%
  buchholz199lumping,%
  dovier04efficient,%
  groote2018prob,%
  wissman2020coalgebra, jansen_2020}.
  
We are interested in the question whether the Paige-Tarjan algorithm
is computationally optimal. A lowerbound for a related problem is
provided in~\cite{berkholz2017tight} that studies colour refinement of
graphs. Colour refinement computes, given a graph and an initial colouring,
a minimal consistent colouring such that every two equally
coloured nodes have, for every colour, the same number of neighbours of
the same colour. More specifically, in that paper it is proven that for
a family of graphs with $n$~nodes and $m$~edges, finding the canonical coarsest
stable colouring is in~$\Omega((m+n) \log n)$. However, the costs for
computations on graphs for colour refinement are charged differently
than those for partition refinement for bisimulation on LTSs. The
former takes edges between blocks of uniformly coloured nodes into
account, the latter focuses on the size of newly created blocks of
states. In~\cite{berkholz2017tight} it is described how the family
of graphs underlying the lowerbound for colour refinement can be
transformed into a family of Kripke structures for which computing
bisimulation is~$\Omega((m+n) \log n)$ when counting the numbers of edges.

In this paper we follow a different approach to obtain a lowerbound.
We define the concept of a partition refinement
algorithm and articulate the complexity in terms of the number of
states that are moved. In particular, we define the notion of a valid
refinement sequence which has its counterpart in iteration sequences
for colour refinement. Then, we introduce a family of (deterministic)
LTSs, called bisplitters, for which we show that computing
bisimulation requires $n \log n$~work.
The family of $n \log n$-hard LTSs that we use to establish the
lowerbound, involves an action set of $\log n$~actions. Building on
this result and exploiting ideas borrowed from~\cite{paige1985linear}
to extend the bisimulation classes for the states in the end
structures, i.e.\ cycles, to the states of the complete LTS, we
provide another family of (deterministic) LTSs that have two actions
only. Then we argue that for the two-action
case the complexity of deciding bisimulation is
$\Omega((m+n) \log n)$. We want to stress that the families involved consist
of deterministic LTSs.

Recently, a linear time algorithm for bisimilarity was proposed
for a PRAM (Parallel Random Access Machine) using $\max(n,m)$
processors~\cite{martens2021linear}.  This algorithm also employs
partition refinement. This naturally raises the question whether the
algorithm is optimal, or whether it can fundamentally be improved. We
answer the question in the present paper by
showing an~$\Omega(n)$ lowerbound for parallel algorithms employing
partition refinement, using a family of deterministic transition
systems with one action label.

We obtain our lowerbound results assuming that algorithms use
partition refinement. However, one may wonder if a different approach
than partition refinement can lead to a faster decision procedure for
bisimulation. For the specific case of deterministic LTSs with a
singleton action set and state labelling,
Robert Paige, Robert Tarjan and Robert Bonic propose a sequential
algorithm~\cite{paige1985linear} that uses linear time. We refer to it
as Roberts' algorithm. In~\cite{castiglione2008hopcroft} it is proven
that partition refinement \`a~la Hopcroft has a lowerbound
of~$\Omega(n \log n)$ in this case.  Concretely, this means that
in the one-letter case Roberts' algorithm achieves the
essentially better performance by using a completely different
technique than partition refinement to determine the bisimulation
equivalence classes.

Crucial for Roberts' algorithm is the ability to identify, in linear
time, the bisimilarity classes of cycles. In this paper we show that
if the alphabet consists of at least two actions, a rapid decision on
`cycles' as in~\cite{paige1985linear} will not be of help to improve
on the Paige-Tarjan algorithm for general LTSs.  We argue that the
speciality in the algorithm of~\cite{paige1985linear}, viz.\ to be able
to quickly decide the bisimilarity of the states on a
  cycle, can be captured by means of a stronger notion, namely an
oracle that provides the bisimulation classes of the states of
a so-called `end structure', the counterpart in the multiple action
setting of a cycle in the single action setting. The oracle can be
consulted to refine the initial partition with respect to the
bisimilarity on the end structures of the LTS for free. We show that
for the class of sequential partition refinement algorithms enhanced
with an oracle as described, thus encompassing the algorithm
of~\cite{paige1985linear}, the $(m{+}n) \log n$~lowerbound persists
for action sets with at least two actions.

For parallel algorithms a similar situation occurs as for
deterministic Kripke structures: an $\bigO(\log n)$ parallel
algorithm exists \cite{jaja1994efficient} to determine the
bisimulation equivalence classes. This algorithm also necessarily
employs techniques that go beyond partition refinement. We believe
that these techniques cannot be used either to fundamentally improve
the complexity of determining bisimilarity on LTSs, but leave the
proof as an open question.

\blankline

\noindent
The document is structured as follows.
In Section~\ref{sec:prelims} we give the necessary preliminaries on
the problem.
A recap of the linear algorithm of~\cite{paige1985linear} is provided
in Section~\ref{sec:roberts}.
Next, we introduce the family of LTSs~$\calB_k$ for
which we show in Section~\ref{sec:hard_example} that deciding
bisimilarity is~$\Omega(n \log n)$ for the class of partition
refinement algorithms and for which we establish in
Section~\ref{sec:oracle} an $\Omega(n \log n)$ lowerbound for the
class of partition refinement algorithms enhanced with an oracle for
end structures.
In Section~\ref{sec:two-action-action-set} we introduce the family of
deterministic LTSs~$\calC_k$, each involving two actions only,  to take the
number of transitions~$m$ into account and establish an $\Omega((m+n)
\log n)$ lowerbound for partition refinement with and without an oracle
for end structures.
In Section~\ref{sec:parallel} we provide the $\Omega(n)$ lowerbound for
parallel refinement algorithms. 
In Section~\ref{sec:color-refinement} we discuss the differences and
similarities with the lowerbound results on colour refinement
of~\cite{berkholz2017tight}.
We wrap up with concluding remarks.

\emph{Note} The present paper is an extension the conference
paper~\cite{GMV21:concur} that appeared in the proceedings of
CONCUR~2021.


\section{Preliminaries}
\label{sec:prelims}

Given a set of states~$S$, a \emph{partition} of~$S$ is a set of
subsets of states $\pi \subseteq 2^S$ such that
$\emptyset \not\in \pi$, for all $B, B' \in \pi$ either 
$B \cap B' = \emptyset$ or~$B = B'$,
and $\bigcup_{B {\in} \pi} \: B = S$. The elements of a partition are
referred to as blocks. A partition~$\pi$ of~$S$ induces an equivalence
relation ${=_\pi} \subseteq {S \times S}$, where for two states
$s, t \in S$, $s =_\pi t$ iff the states $s$ and~$t$ are in the same
block, i.e.\ there is a block $B \in \pi$ such that $s,t \in B$. A
partition~$\pi'$ of~$S$ is a \emph{refinement} of a partition~$\pi$
of~$S$ iff for every block $B' \in \pi'$ there is a block $B \in \pi$
such that $B' \subseteq B$. It follows that each block of~$\pi$ is the
disjoint union of blocks of~$\pi'$. The refinement is
\emph{strict} if $\pi \neq \pi'$.  The common refinement of two
partitions $\pi$ and~$\pi'$ is the partition with blocks
$\lc B \cap B' \mid B \in \pi ,\, B' \in \pi' \colon B \cap B'\neq
\emptyset \rc$.  A sequence of partitions $( \pi_0, \ldots, \pi_n )$
is called a refinement sequence iff $\pi_{i{+}1}$~is a refinement
of~$\pi_i$, for all $0 \leqslant i <n$.

\begin{defi}
  A labelled transition system with initial partition (LTS) is a
  four-tuple
  $L = (S ,\mkern1mu \Act ,\mkern1mu {\rightarrow} ,\mkern1mu \pi_0)$
  where $S$~is a finite set of states~$S$, $\Act$~is a finite alphabet
  of actions, ${\rightarrow} \subseteq {S \times \Act \times S}$ is
  a transition relation, and $\pi_0$~is a partition of the set of
  states~$S$.  A labelled transition system with initial partition is
  called deterministic if the transition relation is a total
  function $S \times \Act \to S$.
\end{defi}

\noindent
Given an LTS $L = (S, \Act, {\rightarrow}, \pi_0)$, states
$s, t \in S$, and an action~$a \in \Act$, we write $s \pijl{a} t$
instead of $(s,a,t) \in {\rightarrow}$.
For notational convenience, we occasionally write
$L[U] = \lc t \in S \mid \exists \mkern1mu s \in U \mkern1mu \exists
\mkern1mu a \in \calA \colon s \pijl{a} t \rc$, and, for a
deterministic LTS~$L$, we may use $L(s,a)$ to denote the unique
state~$t$ of~$L$ such that $s \pijl{a} t$.
We say that $s$~\emph{reaches}~$t$ via~$a$
iff $s \pijl{a} t$.  A state~$s$ reaches a set $U \subseteq S$ via
action~$a$ iff there is a state in~$U$ that is reached by~$s$ via~$a$,
notation $s \pijl{a} U$.
A set of states $V \subseteq S$ is called \textit{stable} under a set
of states $U \subseteq S$ iff for all actions~$a$, either all states
in~$V$ reach~$U$ via~$a$, or no state in~$V$ reaches~$U$ via~$a$.
Thus, a set of states~$V$ is \emph{not} stable under~$U$ iff for two
states $s$ and~$t$ in~$V$ and an action~$a$ it holds that
$s \pijl{a} U$ and $t \stackrel{a}{\nrightarrow} U$. 
A partition~$\pi$ is stable under a set of states~$U$ iff each block
$B \in \pi$ is stable under~$U$. A partition~$\pi$ is called stable
iff it is stable under all its blocks.
So, for any two blocks $B$ and~$C$ of~$\pi$ and any action
$a \in \Act$, either each state~$s$ of~$B$ has an $a$-transition
to~$C$ or each state~$s$ of~$B$ doesn't have an $a$-transition to~$C$.

Following~\cite{park1981gi,milner1980calculus}, given an LTS~$L$, a
symmetric relation~$R \subseteq {S \times S}$ is called a bisimulation
relation iff for all~$(s,t)\in R$ and~$a \in \calA$, we have that
$s \pijl{a} s'$ for some~$s' \in S$ implies that $t \pijl{a} t'$ for
some~$t' \in S$ such that~$(s', t')\in R$.
In the setting of the present paper, as we incorporate the initial
partition in the definition of an LTS, bisimilarity is slightly
non-standard. For a bisimulation relation~$R$, we additionally require
that it respects the initial partition~$\pi_0$ of~$L$, i.e.\
$(s,t)\in R$ implies~$s =_{\pi_0} t$.
Two states~$s,t \in S$ are called (strongly) bisimilar for~$L$ iff a
bisimulation relation~$R$ exists with~$(s,t)\in R$, notation $s \bis_{\!\!L}
t$.
Bisimilarity is an equivalence relation on the set of states
of~$L$. We write~$[s]^{\superbis}_L$ for the bisimulation equivalence
class of the state~$s$ in~$L$.

Note that for a deterministic LTS with a set of states~$S$ and initial
partition~$\pi_0 = \{ S \}$, we have that $\pi_0$~itself already
represents bisimilarity, contrary to LTSs in general.

\blankline

\noindent
Partition refinement algorithms for deciding bisimilarity on LTSs
start with an initial partition~$\pi_0$, which is subsequently
repeatedly refined until a stable partition is reached. Thus,
unstable blocks are replaced by several smaller blocks. The
stable partition that is reached happens to be the coarsest stable
partition of the LTS refining~$\pi_0$ and coincides with
bisimilarity~\cite{kanellakis1983,paige1987three}.

Below we define so-called \emph{valid} refinement sequences.  An
algorithm is called a partition refinement algorithm iff every run of
the algorithm is reflected by a valid refinement sequence
$(\pi_0 , \ldots, \pi_n)$.  All the lowerbounds that we provide apply
to algorithms producing valid partition sequences, which virtually all
known bisimulation algorithms do, and as such this is the core
definition in this paper.

A partition sequence $(\pi_0 , \ldots, \pi_n)$ is valid when the
direct successor~$\pi_i$ of a partition~$\pi_{i-1}$ in the sequence is
obtained by splitting one or more unstable blocks in~$\pi_{i-1}$ using
only information available in $\pi_{i-1}$.  Furthermore, the last
partition in the sequence, the partition~$\pi_n$, is stable.  If
block~$B$ of~$\pi_{i-1}$ is replaced in~$\pi_i$ because it is not
stable under block~$B'$ of~$\pi_{i-1}$, then $B'$ is referred to as a
splitter block.

\begin{defi}
  \label{def:bis_valid}
  Let $L =(S, \Act, {\rightarrow}, \pi_0)$ be an LTS, and $\pi$ a
  partition of~$S$. A refinement~$\pi'$ of~$\pi$ is called a
  \emph{valid refinement} with respect to~$L$ iff the following
  criteria hold. 
  \begin{itemize}[align=parleft, labelsep=0.5cm]
  \item [(a)] $\pi'$ is a strict refinement of $\pi$.
  \item [(b)] \label{case2} If $s \neq_{\pi'}t$ for $s, t \in S$, then
    (i)~$s \neq_{\pi} t$ or (ii)~$s'\in S$ exists such that
    $s \pijl{a} s'$ for some~$a \in \Act$ and, for all~$t' \in S$ such
    that~$t \pijl{a} t'$, it holds that~$s' \neq_{\pi} t'$, or the
    other way around with $t$ replacing~$s$.
  \end{itemize}
  A sequence of partitions $\Pi = ( \pi_0, \ldots, \pi_n )$ is called
  a valid partition sequence iff every successive partition~$\pi_{i}$,
  for $0 < i \leqslant n$, is a valid refinement of~$\pi_{i{-}1}$,
  and, moreover, the partition~$\pi_n$ is stable.
\end{defi}	

\noindent
When a partition~$\pi$ is refined into a partition~$\pi'$, 
states that are in the same block but can reach different blocks can
lead to a split of the block into smaller subsets, say
$k$~subsets. This means that a block $B \in \pi$ is split into
blocks $B_1, \ldots, B_k \in \pi'$. The least amount
of work is done for this operation if we create new blocks for the
least number of states.
That means if $B\in \pi$ is split into $B_1,\ldots, B_k \in \pi'$
and $B_1$ is the biggest block, then the states of $B_2,\ldots, B_k$
are moved to new blocks and the states of $B_1$ remain
in the current block that was holding $B$.
Therefore, we define the refinement costs~$\IRC$ for
the refinement~$\pi'$ of~$\pi$ by
\begin{displaymath}
  \IRC(\pi, \pi') =
  \textstyle{\sum}_{B {\in} \pi} \:
  \bigl ( \,
  |B| -
  \textstyle{\max}_{B' {\in} \pi' \colon B' \subseteq
    B} \: |B'|
  \, \bigr )
  \mkern1mu .
\end{displaymath}
For a sequence of refinements $\Pi = ( \pi_0, \ldots, \pi_n )$ we
write $\IRC(\Pi)$ for $\sum_{i=1}^{n} \: \IRC(\pi_{i{-}1}, \pi_{i})$.
For an LTS~$L$, we have
\begin{displaymath}
  \IRC(L) = \min \lc \IRC(\Pi) \mid
  \text{$\Pi$ a valid refinement sequence for~$L$} \rc
  \mkern1mu .
\end{displaymath}
Note that this complexity measure is different from the one used
in~\cite{berkholz2017tight}, which counts transitions. Our complexity
measure~$\IRC$ is bounded from above by the former.

\blankline

\noindent
In various examples below we characterise the states of LTSs by
sequences of bits.  The set of bits is denoted as
$\Bit = \{ \bit0, \bit1 \}$. Bit sequences of length up to and
including~$k$ are written as $\Bit^{{\leqslant} k}$. The
complement of a bit~$b$ is denoted
by~$\overline{b}$. Thus $\overline{\bit0} = \bit1$ and
$\overline{\bit1} = \bit0$. 
For two bit sequences $\sigma, \sigma'$, we write
$\sigma \preccurlyeq \sigma'$ to indicate that $\sigma$ is a prefix of
$\sigma'$ and write $\sigma \prec \sigma'$ iff $\sigma$ is a strict
prefix of $\sigma'$. For a bit sequence $\sigma \in \Bit^k$, for any
$i,j \leqslant k$, we write~$\sigma[i]$ to indicate the bit at
position~$i$ starting from position $1$. We write
$\sigma[i{:}j] = \sigma[i] \sigma[i{+}1] \cdots \sigma[j]$ to indicate
the subword from position~$i$ to position~$j$. 
Occasionally we use, for a bit sequence~$\sigma$, the
notation~$\sigma\Bit^k$ to denote
$\lc \sigma \sigma' \mid \sigma' \in \Bit^k \rc$, the set of all bit
sequences of length $|\sigma| + k$ having~$\sigma$ as prefix.


\section{Roberts' algorithm}
\label{sec:roberts}

Most algorithms to determine bisimulation for an LTS use partition
refinement. However, there are a few notable exceptions to this. For the
class of deterministic LTSs that have a singleton action alphabet,
deciding the coarsest stable partition, i.e.\ bisimilarity, requires
linear time only; a linear algorithm is due to Robert Paige, Robert
Tarjan, and Robert Bonic~\cite{paige1985linear}, which we therefore
aptly refer to as Roberts' algorithm.

The algorithm of~\cite{paige1985linear} exploits the specific
structure of a deterministic LTS with one action label. An example of
such a transition system is depicted in
Figure~\ref{fig:larger-roberts-example}, where the action label itself
has been suppressed and the initial partition is indicated by
single/double circled states. In general, a
deterministic LTS with one action label can be characterised as a
directed graph, possibly with self-loops, consisting of a number of
cycles of one or more states together with
root-directed trees with their root on a
cycle.  Below we refer to a cycle with the trees connected to it as an
end structure. In a deterministic LTS with one action label, each
state belongs to a unique end structure; it is on a cycle or has a
unique directed path leading to a cycle.

\begin{figure}[htb]
  \centering

\scalebox{0.85}{%
\begin{tikzpicture}

  \node [state, accepting] (c1) at (12.5,1.25) {$c_1$} ;
  \node [state] (c2) at (12.25,3) {$c_2$} ;
  \node [state, accepting] (c3) at (9.75,4) {$c_3$} ;
  \node [state, accepting] (c4) at ( 7.5,3) {$c_4$} ;
  \node [state] (c5) at (7.75,0.75) {$c_5$} ;
  \node [state, accepting] (c6) at (10.25,0) {$c_6$} ;

  \draw (c1) edge [out=60, in=-60] (c2) ;
  \draw (c2) edge [out=120, in=0] (c3) ;
  \draw (c3) edge [out=180, in=60] (c4) ;
  \draw (c4) edge [out=-120, in=120] (c5) ;
  \draw (c5) edge [out=-60, in=180] (c6) ;
  \draw (c6) edge [out=0, in=-120] (c1) ;

  \node [state, accepting] (s11) at (14.5,1.5) {$s_{11}$} ;
  \node [state] (s12) at (16.0,2.5) {$s_{12}$} ;
  \node [state, accepting] (s13) at (17.0,4.0) {$s_{13}$} ;
  \node [state, accepting] (s14) at (17.5,5.75) {$s_{14}$} ;
  \node [draw=none] (T1) at (18,6.25) {$T_1$} ;
  
  \draw (s14) edge [bend left=10] (s13) ;
  \draw (s13) edge [bend left=10] (s12) ;
  \draw (s12) edge [bend left=10] (s11) ;
  \draw (s11) edge [bend left=10] (c1) ;

  \node [state, accepting] (s21) at (14,5) {$s_{21}$} ;
  \node [state, accepting] (s22) at (14.75,7.25) {$s_{22}$} ;
  \node [state, accepting] (s23) at (12.75,7.5) {$s_{23}$} ;
  \node [draw=none] (T2) at (13.75,7.875) {$T_2$} ;
  
  \draw (s23) edge [bend left=10] (s21) ;
  \draw (s22) edge [bend left=10] (s21) ;
  \draw (s21) edge [bend left=20] (c2) ;

  \node [state] (s32) at (10.25,6.25) {$s_{32}$} ;
  \node [state, accepting] (s31) at (8.75,5.75) {$s_{31}$} ;
  \node [draw=none] (T3) at (9.25,6.5) {$T_3$} ;

  \draw (s32) edge [bend left=10] (c3) ;
  \draw (s31) edge [bend right=10] (c3) ;
  
  \node [state] (s41) at (5.75,4.75) {$s_{41}$} ;
  \node [state, accepting] (s42) at (7,6.75) {$s_{42}$} ;
  \node [state, accepting] (s43) at (5,6.5) {$s_{43}$} ;
  \node [state, accepting] (s44) at (6,7.75) {$s_{44}$} ;
  \node [draw=none] (T4) at (4.5,7) {$T_4$} ;

  \draw (s44) edge [bend right=10] (s41) ;
  \draw (s43) edge [bend right=10] (s41) ;
  \draw (s42) edge [bend right=10] (s41) ;
  \draw (s41) edge [bend right=10] (c4) ;

  \node [state] (s51) at (5.75,0.25) {$s_{51}$} ;
  \node [state, accepting] (s52) at (5.25,1.75) {$s_{52}$} ;
  \node [state] (s53) at (4,3.5) {$s_{53}$} ;
  \node [draw=none] (T5) at (3.5,4) {$T_5$} ;

  \draw (s53) edge [bend right=10] (s52) ;
  \draw (s52) edge [bend right=10] (c5) ;
  \draw (s51) edge [bend right=10] (c5) ;
\end{tikzpicture}
} 
  \caption{An example of a deterministic LTS with initial partition (action
    label suppressed).}
  \label{fig:larger-roberts-example} 
\end{figure}
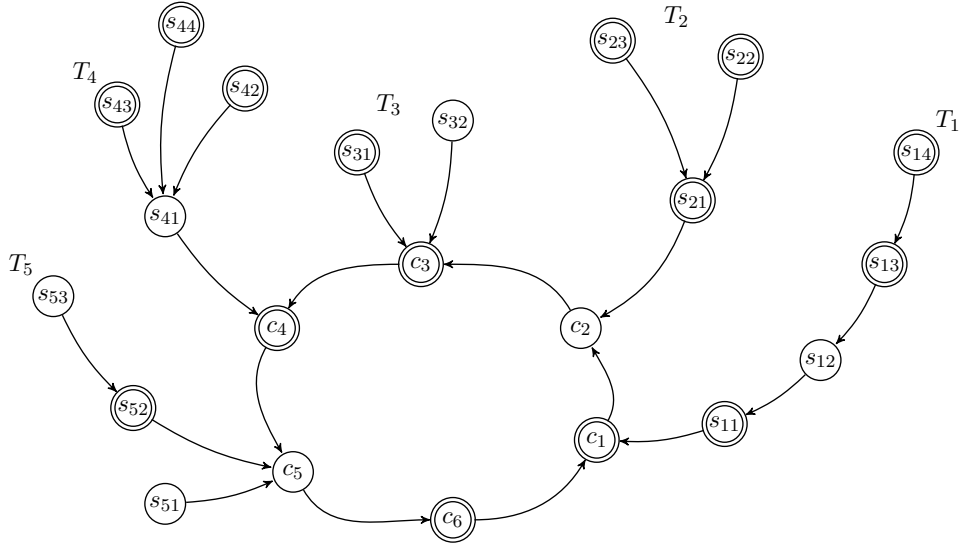

In brief, Roberts' algorithm for deterministic LTSs
with one action label can be described as follows
(see~\cite{paige1985linear} for more details).

\begin{enumerate}
\item As preparatory step, find all the end structures of the LTS,
  i.e.\ detect all cycles, and all
  root-oriented trees leading to cycles.
  
\item Observe that each state~$s$ on a cycle encodes a sequence of
  blocks, viz.\ the sequence starting from the block the state is in,
  and blocks encountered when following the transitions, up to the state 
  on the cycle that leads back to~$s$. This sequence of blocks forms 
  a word~$w$ over the alphabet of the initial partition, where each
  block of the initial partition is a symbol of this alphabet. The
  word~$w$ can be uniquely written as $v^{k}$ with~$v$
  of minimal length and~$k > 0$. The string~$v$ is
  referred to as the repeating prefix of the state~$s$.

  We consider the repeating prefixes of all states on the cycle and
  identify the lexicographically least repeating prefix~$v$.  This can
  be done in linear time in the size of the cycle using a string
  matching algorithm due to Knuth, Morris, and
  Pratt~\cite{KnuthMP77}. The lexicographically least repeating
  prefix~$v$ and the minimal number of transitions that is required to
  reach a state~$t$ from a state that has $v$ as repeating prefix,
  determines the bisimulation equivalence class of the state~$t$. We
  encode this bisimulation equivalence class by the corresponding
  rotation of the prefix~$v$.  This way the bisimulation class is
  established for all states on all cycles. By comparing least
  repeating prefixes bisimilarity across cycles can be detected.

\item By a backward calculation along the path leading from a state up
  in a tree down to their root on a cycle, the bisimilarity
  equivalence classes for the remaining states can subsequently be
  determined in linear time as well. The root of the tree is a state
  on the cycle and therefore has been assigned a string, hence a
  bisimulation class. We assign to a child the string of the parent
  prepended with the symbol of the initial class of the child.
\end{enumerate}

\noindent
\paragraph{Example}
The deterministic LTS of Figure~\ref{fig:larger-roberts-example} has a
single end structure, viz.\ the cycle formed by the states $c_1$
to~$c_6$ and five trees, tree~$T_1$ with leaf~$s_{14}$ and rooted
in~$c_1$, tree~$T_2$ with leaves $s_{22}$ and~$s_{23}$ rooted
in~$c_2$, tree~$T_3$ with leaves $s_{31}$ and~$s_{32}$ rooted
in~$c_3$, the tree~$T_4$ with leaves $s_{42}$, $s_{43}$, and~$s_{44}$
rooted in~$c_4$, and the tree~$T_5$ with leaves $s_{51}$ and~$s_{53}$
rooted in~$c_5$.

With the symbol~$A$ for an accepting, i.e.\ double-circled, state and
the symbol~$N$ for a non-accepting state, i.e.\ not double-circled,
we associate the following sequences of $A$'s and~$N$'s to the nodes
on the cycle:
\begin{displaymath}
  \begin{array}{l@{\,}c@{\,}lcl@{\,}c@{\,}lcl@{\,}c@{\,}l}
    c_1 & : & \mathit{ANAANA} && c_3 & : & \mathit{AANAAN} && c_5 & :
    & \mathit{NAANAA} \\ 
    c_2 & : & \mathit{NAANAA} && c_4 & : & \mathit{ANAANA} && c_6 & :
    & \mathit{AANAAN} \mkern1mu . \\ 
  \end{array}
\end{displaymath}
With $A$ preceding~$N$, the lexicographically least repeating prefix
is~$\mathit{AAN}$. We assign states~$c_1$ and~$c_4$ to the
bisimulation class of~$\mathit{ANA}$, states~$c_2$ and~$c_5$ to the
bisimulation class of~$\mathit{NAA}$, and states~$c_3$ and~$c_6$ to
the bisimulation class of~$\mathit{AAN}$. Here, $\mathit{ANA}$
and~$\mathit{NAA}$ are the 1-place and 2-place rotations
of~$\mathit{AAN}$, respectively.

Moving to tree~$T_1$ with root~$c_1$ having
string $\mathit{ANA}$ associated with it, we concatenate for
state~$s_{11}$ the symbol~$A$, since $s_{11}$ has been assigned in the
block of accepting states initially, followed by the
string~$\mathit{ANA}$ of~$c_1$, forming $\mathit{AANA}$ which is
reduced to~$\mathit{AAN}$ (exploiting the equality
$A(\mathit{ANA})^\omega = (\mathit{AAN})^\omega)$. Thus, we see that
the states~$s_{11}$ and~$c_6$ are bisimilar. Similarly, for
state~$s_{12}$ we prepend the symbol~$N$ of the child~$s_{12}$ to the
string $\mathit{AAN}$ of the parent~$s_{11}$ and obtain
$\mathit{NAAN}$ which is reduced to~$\mathit{NAA}$ (now exploiting the
equality $N(\mathit{AAN})^\omega = (\mathit{NAA})^\omega)$, as for
state~$c_5$.  Considering in contrast the state~$s_{41}$, with
symbol~$N$ and which is child of state~$c_4$ with
string~$\mathit{ANA}$, $s_{41}$~gets assigned the
string~$\mathit{NANA}$ (which can not be reduced, since
$N(\mathit{ANA})^\omega \neq (\mathit{NAN})^{\omega}$). Subsequently,
the states $s_{42}$ to~$s_{44}$ get assigned $\mathit{ANANA}$.

\blankline

\noindent
Roberts' algorithm solves the so-called single function coarsest
partition problem in~$\bigo{n}$ for a set of $n$ elements.  A
striking result is that any algorithm that is based on partition
refinement requires $\Omega(n \log n)$, as witnessed
in~\cite{berstel2004complexity,castiglione2008hopcroft}, where it is
shown that partition refinement algorithm of
Hopcroft~\cite{hopcroft1971DFAmin} cannot do better than
$\bigo{n \log n}$. Thus, Roberts' algorithm must use other techniques
than partition refinement.  Below we come back to this observation,
showing that it is not possible to use the ideas in Roberts' algorithm
to come up with a linear algorithm for computing bisimilarity for a
class of LTSs that either includes nondeterministic LTSs, or allows
LTSs to involve more than one action label.
	

\section{$\calB_k$ is $\Omega(n \log n)$ for partition refinement}
\label{sec:hard_example}

In this section we introduce a family of deterministic LTSs called
bisplitters~$\calB_k$ for $k \geqslant 1$, on which the cost of any
partition refinement algorithm is $\Omega(n\log n)$, where $n$~is the
number of states. Building on the family of $\calB_k$'s, we propose in
Section~\ref{sec:two-action-action-set} a family of LTSs $\calC_k$
for which the cost of partition refinement is
$\Omega((n+m)\log n)$, where $m$~is the number of transitions.

\begin{figure}[htp]
  \begin{center}

\scalebox{0.85}{%
\begin{tikzpicture} 

\node [state] (b0) at (0,0) {$\bit{0}$};
\node [state, accepting] (b1) at (2,0) {$\bit{1}$};

\node (B1) at (-1,0.25) {$\calB_1$};
\end{tikzpicture}
} 
\qquad \qquad \qquad
\scalebox{0.85}{%
\begin{tikzpicture} 

\node [state] (b00) at (4,0) {$\bit{00}$};
\node [state] (b01) at (6,0) {$\bit{01}$};
\node [state, accepting] (b10) at (8,0) {$\bit{10}$};
\node [state, accepting] (b11) at (10,0) {$\bit{11}$};

\node at (11,0.25) {$\calB_2$};

\draw
(b00) edge [loop above] node [above] {$a_1$} (b00)
(b01) edge [bend  left] node [above] {$a_1$} (b10)
(b10) edge [loop above] node [above] {$a_1$} (b10)
(b11) edge [out=105, in=75] node [above] {$a_1$} (b00)
;

\end{tikzpicture}
} 

\bigskip

\scalebox{0.85}{%
\begin{tikzpicture}[> = stealth', semithick, node distance=2cm] 


  \node[state] (000) at (0,2.5) {$\bit{000}$};
  \node[state] (001) at (2,2.5) {$\bit{001}$};
  \node[state] (010) at (4,2.5) {$\bit{010}$};
  \node[state] (011) at (6,2.5) {$\bit{011}$};
  \node[state, accepting] (100) at (0,0) {$\bit{100}$};
  \node[state, accepting] (101) at (2,0) {$\bit{101}$};
  \node[state, accepting] (110) at (4,0) {$\bit{110}$};
  \node[state, accepting] (111) at (6,0) {$\bit{111}$};

  \path[->] 
  (000) edge [loop above] node {$a_1, a_2$} (000)
  
  (001) edge [loop above] node {$a_1$} (001)
  (001) edge [bend left]  node [above] {$a_2$} (010)
  
  (010) edge [out=-135, in=45.5] node [pos=0.1375, above left] {$a_1$} (100)
  (010) edge [loop above] node [above] {$a_2$} (010)
  
  (011) edge [out=-135, in=44.5] node [pos=0.1, above left]{$a_1$} (100)
  (011) edge [out=105, in=75, looseness=1] node [above] {$a_2$} (000)
  
  (100) edge [<-,loop below] node{$a_1, a_2$} (100)
  
  (101) edge [bend right] node [below] {$a_2$} (110)
  (101) edge [<-, loop below] node [below] {$a_1$} (101)
  
  (110) edge [out=135, in=-45.5] node [pos=0.1, above right] {$a_1$} (000)
  (110) edge [<-, loop below] node [below] {$a_2$} (110)
  
  (111) edge [out=135, in=-44.5] node[pos=0.075, above right] {$a_1$} (000)
  (111) edge[out=-105, in=-75, looseness=1] node [below] {$a_2$} (100)
   ;

   \node (B3) at (7,-0.25) {$\calB_3$} ;
   \node (x)  at (-1,-0.25) {$\phantom{\calB_3}$} ;
\end{tikzpicture}
}
  \end{center}
  \caption{The bisplitters $\calB_1$, $\calB_2$, and
    $\calB_3$. Initial partitions are indicated by single-circled and double circled states.}
\label{fig:bisplitter12}
\end{figure}
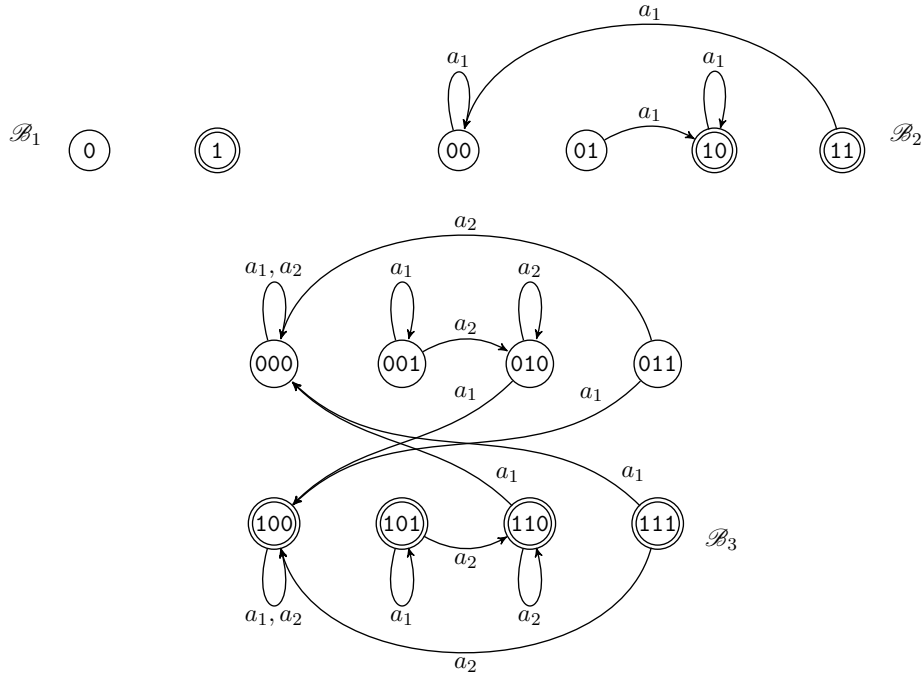

\begin{defi} 
  \label{def:bisplitter}
  For $k \geqslant 1$, the bisplitter~$\calB_k$ is defined as the LTS
  with initial partition
  $\calB_k = (\Bit^{k}, \Act_k, {\rightarrow},
  \pi_0^{k})$ where the set of states~$\Bit^{k}$ is the
  set of all bit strings of length~$k$,
  $\Act_k = \{ a_1, \ldots, a_{k{-}1} \}$ is a set of
  $k{-}1$~actions, the transition relation is given by
  \begin{align*}
    & \lc \sigma \pijl{a_i} \sigma \mid
      \sigma \in \Bit^{k} ,\, 1 \leqslant i < k \colon \sigma[i{+}1] = \bit0
      \rc \cup {} \smallskip \\
    & \qquad
      \lc \sigma \pijl{a_i} \sigma[1{:}i {-} 1] \mkern2mu
      \overline{\sigma[i]} \mkern2mu \bit{0}^{k{-}i} \mid
      \sigma \in \Bit^{k} ,\, 1 \leqslant i < k \colon \sigma[i{+}1] = \bit1
      \rc \mkern1mu ,
  \end{align*}
  and 
  \begin{math}
    \pi_0^{k} =
    \lbrace B_{\bit0} , B_{\bit1} \rbrace \text{, where $B_{\bit0}=\bit0 \Bit^{k-1}$ and $B_{\bit1}=\bit1 \Bit^{k-1}$,}
  \end{math}
  is the initial partition.
\end{defi}

\noindent
We see that the bisplitter~$\calB_k$ has~$2^k$ states, viz.\
all bit strings of length~$k$, and $k{-}1$~different actions. The
LTS~$\calB_k$ is deterministic. Each state has exactly one outgoing
transition for each action~$a_i$, $1 \leqslant i < k$. Thus, $\calB_k$
has $(k{-}1) \mkern0.5mu 2^{k}$~transitions: (i)~a self-loop for
bitstring~$\sigma$ with label~$a_i$ if the $i{+}1$-th
bit~$\sigma[i{+}1]$ of~$\sigma$ equals~$\bit0$; (ii)~otherwise, i.e.\
when bit~$\sigma[i{+}1]$ equals~$\bit1$, the bitstring~$\sigma$ has
a transition for label~$a_i$ to the
bitstring that equals the first $i{-}1$ bits of~$\sigma$, flips the
$i$-th bit of~$\sigma$, and has $k{-}i$ many~$\bit0$'s following. The
initial partition~$\pi_0^{k}$ distinguishes the bit strings
starting with~$\bit0$ from those starting with~$\bit1$.

Drawings of the first three bisplitters $\calB_1$ to~$\calB_3$ are
given in Figure~\ref{fig:bisplitter12}. We see in the picture
of~$\calB_3$ for example for the
bitstring~$\sigma = \bit1 \bit0 \bit1$, an $a_1$-transition to itself,
as $\sigma[1{+}1] = \sigma[2] = \bit0$ and an $a_2$-transition to the
bitstring $\bit1 \bit1 \bit0$, as $\sigma[2{+}1] = \sigma[3] = \bit1$,
$\overline{\sigma[2]} = \overline{\bit0} = \bit1$, and
$\sigma[1] \overline{\sigma[2]} \bit0 = \bit{110}$. As another
illustration, for the bitstring $\sigma = \bit{011}$ we have an
$a_1$-transition to $\overline{\sigma[1]} \bit{00} = \bit{100}$ since
$\sigma[1{+}1] = \bit1$ and an $a_2$-transition to
$\sigma[1] \overline{\sigma[2]} \bit0 = \bit{000}$ since
$\sigma[2{+}1] = \bit1$.

Also note that $\calB_3$ contains two copies of $\calB_2$. In the
copies, the action label~$a_1$ of~$\calB_2$ maps to the action
label~$a_2$ in~$\calB_3$, and each state associated with a
bitstring $\sigma \in \Bit^{2}$ produces two copies in~$\calB_3$; one
copy is obtained by the
mapping~$\sigma \mapsto \bit0 \mkern0.5mu \sigma$ and the other copy
is obtained by the mapping $\sigma \mapsto \bit1\mkern0.5mu \sigma$.
In general, bisplitter~$\calB_k$ is twice embedded in
bisplitter~$\calB_{k{+}1}$ via the mappings
$\sigma \mapsto \bit0 \mkern0.5mu \sigma$ and
$\sigma \mapsto \bit1\mkern0.5mu \sigma$ from $\Bit^{k}$
to~$\Bit^{k{+}1}$ for the states using the mapping
$a_i \mapsto a_{i{+}1}$ from $\calA_k$ to~$\calA_{k{+}1}$ for the
action labels. Note that initial partitions are not respected.

\begin{defi}
  For any string~$\sigma \in \Bit^{{\leqslant}k}$, we define the
  prefix block~$B_{\sigma}$ of~$\calB_k$ to be the block
  $B_{\sigma} = \lc \sigma'\in \Bit^{k} \mid \sigma \preccurlyeq \sigma' \rc$.
\end{defi}

\noindent
The following lemma collects a number of results related to prefix
blocks that we need in our complexity analysis for computing
bisimilarity for the bisplitters.

\begin{lem}
  \label{thm:facts}
  Let $k \geqslant 1$ and consider the LTS with initial
  partition~$\calB_k = (\Bit^k, \Act_k, {\rightarrow},
  \pi_0^{k})$, i.e.\
  the $k$-th bisplitter. Let the sequence
  $\Pi = ( \pi_0^{k}, \ldots, \pi_n )$ be a valid refinement
  sequence for~$\calB_k$. Then it holds that
  \begin{enumerate}[(a)]
  \item \label{thm:facts:1} Every partition $\pi_i$ in~$\mkern1mu \Pi$
    contains prefix blocks only.
  \item \label{thm:facts:2} If partition~$\pi_i$ in~$\mkern1mu \Pi$
    contains a prefix block~$B_\sigma$ with $|\sigma| < k$, then
    $\pi_i$ is not stable.
  \item \label{thm:facts:3} If $B_{\sigma}$ is
    in~$\pi_{i}$, for $0 \leqslant i < n$, then either
    $B_{\sigma} \in \pi_{i{+}1}$, or
    $B_{\sigma \bit{1}} \in \pi_{i{+}1}$
    and~$B_{\sigma \bit{0}} \in \pi_{i{+}1}$.
  \end{enumerate}
\end{lem}
      
\begin{proof}
  \ref{thm:facts:1} Initially, for
  $\pi_0^{k} = \{ B_{\bit{0}}, B_{\bit{1}} \}$, both
  its blocks are prefix blocks by definition.
  We prove, if partition~$\pi_i$, for $0 \leqslant i < n$, consists
  of prefix blocks only, then all blocks in~$\pi_{i+1}$ are
  prefix blocks as well.

  Assume, to arrive at a contradiction, that there is a block
  $B \in \pi_{i+1}$ that is not a prefix block. Because $\pi_{i+1}$ is
  a refinement of $\pi_i$, we have $B \subseteq B_{\sigma}$ for some
  prefix block $B_{\sigma} \in \pi_i$. This means that $\sigma$ is a
  common prefix of all elements of $B$. We can choose $\theta$ such
  that $\sigma\theta$ is the longest common prefix of all elements of
  $B$. Since every singleton of $\mathbb{B}^k$ is a prefix block,
  $B$~is not a singleton. This implies that $|\sigma\theta| < k$ and
  that there are elements $\sigma_1$ and~$\sigma_2$ of~$B$ such that
  $\sigma\theta \bit{0}$ is a prefix of~$\sigma_1$ and
  $\sigma\theta 1$~is a prefix of~$\sigma_2$. Because $B$~is not a
  prefix block by assumption, there must exist a
  string~$\tau \in \mathbb{B}^k$ with prefix~$\sigma\theta$ such
  that $\tau \not\in B$.  Obviously, we have either
  (i)~$\sigma\theta \mkern0.5mu \bit0$ is a prefix of~$\tau$, or
  (ii)~$\sigma\theta \mkern0.5mu \bit1$ is a prefix of~$\tau$.  We
  will show that in both these cases $\tau$ in fact belongs to~$B$,
  thus arriving at a contradiction.
  \begin{enumerate}[(i)]
  \item Suppose $\sigma\theta \mkern0.5mu \bit0$ is a prefix
    of~$\tau$. We argue that $\tau$ and~$\sigma_1$
    belong to the same block in~$\pi_{i+1}$ 
    since, for each~$a_j$, $1 \leqslant j < k$, the target states
    $\sigma_1'$ and~$\tau'$ of the transitions $\sigma_1 \pijl{a_j} \sigma_1'$
    and $\tau \pijl{a_j} \tau'$ belong to the same block
    of~$\pi_i$. There are three cases:
    \begin{itemize}
    \item $j < |\sigma\theta|$: Since $\sigma\theta$ is a prefix of
      both $\sigma_1$ and $\tau$, we have $\sigma_1[j{+}1] = \tau[j{+}1]$.
      \begin{itemize}
      \item If $\sigma_1[j{+}1] = \tau[j{+}1] = \bit{0}$, then
        $\sigma_1' = \sigma_1$ and $\tau' = \tau$.  Obviously, both
        $\sigma_1'$ and $\tau'$ belong to~$B_{\sigma}$ (since
        $\sigma_1$ and $\tau$ belong to~$B_{\sigma}$).
      \item If $\sigma_1[j{+}1] = \tau[j{+}1] = \bit{1}$, then both $\sigma_1'$
        and $\tau'$ are of the form
        $\varrho[1{:}j{-}1] \mkern0.5mu \overline{\varrho[j]} \mkern0.5mu
        \bit0^{k-j}$ where $\varrho=\sigma\theta$, and we have
        $\sigma_1' = \tau'$. So, they clearly belong to the same
        block of~$\pi_i$.
      \end{itemize}
    \item $j = |\sigma\theta|$: Since
      $\sigma_1[j{+}1] = \tau[j{+}1] = \bit{0}$, we have
      $\sigma_1' = \sigma_1$ and $\tau' = \tau$, and
      hence both $\sigma_1'$ and $\tau'$ belong to
      $B_{\sigma}$.
    \item $j > |\sigma\theta|$:
        In this case, for a string of the
        form~$\sigma \theta \varrho$, an $a_j$-transition leads to a
        string of the form~$\sigma \theta \varrho'$. In particular
      this means that if $j > |\sigma\theta|$ and
      $\sigma_1 \pijl{a_j} \sigma_1'$ and $\tau \pijl{a_j} \tau'$,
      then $\sigma\theta$ is a prefix of both $\sigma_1'$ and $\tau'$,
      and $\sigma_1'$ and $\tau'$ belong to $B_{\sigma}$ in $\pi_i$.
    \end{itemize}
  \item Now, suppose $\sigma\theta \bit{1}$ is a prefix of~$\tau$. We
    argue that $\tau$ and~$\sigma_2$ belong to the
    same block in~$\pi_{i+1}$ because for each $a_j$ (where
    $1 \leqslant j < k$) the transitions $\sigma_2 \pijl{a_j} \sigma_2'$ and
      $\tau \pijl{a_j} \tau'$ lead to the same block
    of~$\pi_i$. Also, here there are three cases:
    \begin{itemize}
    \item $j < |\sigma\theta|$: Similar as for (i).
    \item $j = |\sigma\theta|$: Since $\sigma_1[j{+}1] = \tau[j{+}1] = \bit{1}$,
      we have
      $\sigma_1' = \tau' = \varrho[1{:}j{-}1]
      \mkern0.5mu \overline{\varrho[j]} \mkern0.5mu \bit0^{k{-}1}$ where
      $\varrho = \sigma \theta$, so clearly $\sigma_1'$ and $\tau'$
      are in a same block in~$\pi_i$.
    \item $j > |\sigma\theta|$: Similar as for~(i).
    \end{itemize}
    Thus, both in case (i) and in case~(ii) we see that we must have
    $\tau \in B$, contradicting the choice for~$\tau$.
  \end{enumerate}

\noindent
  \ref{thm:facts:2}
  Suppose $B_\sigma \in \pi_i$ and $| \sigma |\ = \ell < k$. Let
  $\theta \in \Bit^\ast$ be such that
  $\sigma_1 = \sigma \bit0 \theta$ and~$\sigma_2 = \sigma \bit1
  \theta$. Then we have $\sigma_1 \pijl{a_\ell} \sigma_1 \in B_\sigma$ and
  $\sigma_2 \pijl{a_\ell} \sigma[1{:}\ell{-}1] \mkern1mu
  \overline{\sigma[\ell]} \mkern1mu {\bit0}^{k{-}\ell} \notin
  B_\sigma$. Thus $B_\sigma$ isn't stable, and hence~$\pi_i$ isn't
  either.

  \medskip

  \noindent
  \ref{thm:facts:3} We show that for a prefix block
  $B_\sigma\in\pi_i$, a bit $\bit{b} \in \Bit$ and all
  $\theta, \theta' \in \Bit^{k{-}(|\sigma|{+}1)}$, the states
  $\sigma_1 = \sigma \mkern1mu \bit{b} \mkern1mu \theta$ and
  $\sigma_2 = \sigma \mkern1mu \bit{b} \mkern1mu \theta'$ are not
  split by action~$a_j$, for $1 \leqslant j < k$, and thus are in the
  same block of~$\pi_{i{+}1}$.
  Pick~$j$, $1 \leqslant j < k$, and suppose
  $\sigma_1 \pijl{a_j} \sigma'_1$~and
  $\sigma_2 \pijl{a_j} \sigma'_2$. If $j \leqslant |\sigma|$
  and $\sigma[j+1] = \bit0$ then $\sigma_1' = \sigma_1$
  and~$\sigma'_2 = \sigma_2$, hence both
  $\sigma'_1, \sigma'_2 \in B_\sigma$ don't split for~$a_j$. If
  $j \leqslant |\sigma|$ and $\sigma[j+1] = \bit1$ then
  $\sigma'_1 = \sigma'_2$ and don't split for~$a_j$ either.
  If~$j > |\sigma|$ then both $\sigma'_1, \sigma'_2 \in B_\sigma$ and
  don't split for~$a_j$ either.
\end{proof}

\noindent
With the help of the above lemma, clarifying the form of the
partitions in a valid refinement sequence for the bisplitter family,
we are able to obtain a lowerbound for any algorithm exploiting
partition refinement to compute bisimilarity.
\begin{thm}
  \label{thm:inherent_complexity}
  For any $k > 1$, application of partition refinement to the
  bisplitter~$\calB_k$ has refinement costs
  $\IRC(\calB_k) \in \Omega(n \log n)$ where $n = 2^k$ is the number
  of states of~$\calB_k$.
\end{thm}

\begin{proof}
  Let $\Pi = ( \pi_0^{k}, \ldots, \pi_m )$ be a valid
  refinement sequence for~$\calB_k$. By items \ref{thm:facts:1}
  and~\ref{thm:facts:2} of Lemma~\ref{thm:facts}, we have
  $\pi_m = \lc \{s\} \mid s \in \Bit^k \rc$ since $\pi_m$~is stable
  and thus contains singleton blocks only. Item~\ref{thm:facts:3} of
  Lemma~\ref{thm:facts} implies that in every refinement
  step~$(\pi_i, \pi_{i{+}1})$ a block is either kept or it is refined
  in two prefix blocks of equal size. The cost of refining the
  block~$B_\sigma$, for $1 \leqslant |\sigma| \leqslant k{-}1$,
  into~$B_{\sigma\bit{0}}$ and $B_{\sigma\bit{1}}$ is the number of
  states in~$B_{\sigma\bit{0}}$ or the number of states 
  in~$B_{\sigma\bit{1}}$, which are the same and are equal to
  $\frac{1}{2} 2^{\mkern1mu k{-}|\sigma|}$.
  Therefore, we have
  \begin{displaymath}
    \IRC(\Pi) = \sum_{\ell=1}^{k{-}1} \: 2^\ell \frac{1}{2}
    2^{k-\ell} =
    \sum_{\ell=1}^{k{-}1} \: \frac{1}{2} 2^{k} =
    (k{-}1) 2^{k{-}1} \mkern1mu .
  \end{displaymath}
  With $n$~the number of states of~$\calB_k$, we have that $n=2^k$,
  thus $k{-}1 = \log \frac 12 n$. Hence,
  $\IRC(\Pi) = \frac{1}{2} n \log \frac12 n$ which is in
  $\Omega(n\log n)$.

  Thus, for every valid partition refinement sequence~$\Pi$
  for~$\calB_k$ we have $\IRC(\Pi) \in \Omega(n \log n)$. In particular this bound applies to
  the valid refinement sequence of minimal cost, and hence we conclude $\IRC(\calB_k) \in \Omega{(n \log n)}$.
\end{proof}


\section{$\calB_k$ is $\Omega(n\log n)$ for partition refinement with
  an oracle}
\label{sec:oracle}

In the previous section we have shown that computing bisimilarity with
partition refinement for the family of bisplitters is
$\Omega(n \log n)$. The bisplitters are deterministic LTSs but have
growing actions sets. For the corner case of deterministic transition
systems with a singleton action set, Roberts' algorithm
discussed in Section~\ref{sec:roberts} establishes
bisimilarity in~$\bigo{n}$. Linearity was obtained by the trick of
calculating the (lexicographically) least repeating prefix on the
cycles in the transition system.

One may wonder whether an approach different from partition refinement
of establishing bisimulation equivalence classes for transition
systems with non-degenerate action sets can provide a linear
performance. In order to capture the approach
of~\cite{paige1985linear}, we augment the class of partition
refinement algorithms with an oracle. At the start of the algorithm
the oracle can be consulted to identify the bisimulation classes for
designated states, viz.\ for those that are in a so-called \emph{end
  structure}, the counterpart of the cycles in Roberts' algorithm. 
This results in a refinement of the initial partition;
partition refinement then starts from the updated partition.

Thus, we can ask the oracle to provide the bisimulation classes of all
states in an end structure of the input LTS, also including
bisimilar states of the LTS not in an end structure. This yields a new
partition, viz.\ the common refinement of the initial partition, on
the one side, and the partition induced by the bisimulation
equivalence classes as given by the oracle and the complement of their
union, on the other side. Hence, the work that remains to be done is
establishing the bisimulation equivalence classes, with respect to the
initial partition, for the states not bisimilar to any in an end
structure.

We will establish that a partition refinement algorithm that can
consult an oracle cannot improve upon the complexity
of computing bisimulation by partition refinement. We first define
the notion of an end structure of an LTS formally as well as the
associated notion of an end structure partition.

\begin{defi}
  \label{def-end-structure}
  Given an LTS $L = (S, \Act, {\rightarrow}, \pi_0)$ with an initial
  partition, a non-empty subset~$S' \subseteq S$ is called an
  \emph{end structure} of~$L$ iff $S'$~is a minimal set of states
  that is closed under all transitions, $L[S'] \subseteq S'$
  and for all~$S'' \subseteq S$ it holds that
  $L[S''] \subseteq S''$ and $S'\cap S'' \neq \emptyset$ implies 
  $S' \subseteq S''$. Moreover, $\ES(L) = \lc S' \subseteq S \mid
  \text{$S'$ end structure of~$L$}\rc$, $\mathit{ES}(L) = \bigcup \ES(L)$,
  and the partition~$\piES$ such that
  \begin{displaymath}
    \piES = \lc [s]^{\superbis}_L \mid
    s \in \mathit{ES}(L) \rc \cup
    \lc B \setminus
    \textstyle{\bigcup_{s {\in} \mathit{ES}(L)}} \: [s]^{\superbis}_L \mid
    B \in \pi_0 \rc
    \setminus \{ \emptyset \}
  \end{displaymath}
  is called the end structure partition of~$L$.
\end{defi}

\noindent
Like the cycles exploited in Roberts' algorithm, an LTS can have
multiple end structures. The end structure partition~$\piES$ consists
of all the bisimilarity equivalence classes of~$L$ that include at
least one state of an end structure, completed with blocks
holding the remaining states, if non-empty. So, for every state~$s$ of
an end structure, the end structure partition has identified all
states that are bisimilar to state~$s$ and separates $s$ and its bisimilar states from the rest of the LTS\@. The other states are assigned in
the end structure partition to the blocks just as the initial
partition does.

\noindent
\paragraph{Example} In the LTS of
Figure~\ref{fig:larger-roberts-example} the cycle of $c_1$ to~$c_6$ is
the only end structure.
All states have a path to the cycle, hence every non-empty set that is
  closed under transitions will contain the cycle. Would the LTS have
  contained any isolated states, these would be end structures by themselves.
Thus, consultation of the oracle leads to the
refinement of the initial structure~$\pi_0$ that consists of the two
blocks
\begin{align*}
  & \{ c_1, c_3, c_4, c_6 \} \cup \{ s_{11} , s_{13}, s_{14} \} \cup \{
    s_{21}, s_{22}, s_{23} \} \cup \{ s_{31} \} \cup \{ s_{42}, s_{43},
    s_{44} \} \cup \{ s_{52} \}
    \smallskip \\
  & \qquad
    \text{and} \ 
    \{ c_2, c_5 \} \cup \{ s_{12} \} \cup \{ s_{32} \} \cup \{ s_{41} \}
    \cup \{ s_{51}, s_{53} \}
    \intertext{into end structure partition $\piES$ with five blocks,
    viz.\ the three blocks} 
  & \{ c_1, c_4 \} \cup \{ s_{13} ,\, s_{21} ,\, s_{52} \}  ,\
    \{ c_2, c_5 \} \cup \{ s_{12} , s_{32} \} ,\ \text{and} \ 
    \{ c_3, c_6 \} \cup \{ s_{11} ,\, s_{14} ,\, s_{22} ,\, s_{23} \}
  \\
  \intertext{consisting of the states on the cycle together with the
  states that are bisimilar,
  on the one hand, and the two blocks with the remaining
  states} 
  & \{  s_{31} \} \cup \{ s_{42}, s_{43}, s_{44} \}
    \  \text{and} \ 
    \{ s_{41} \} \cup \{ s_{51}, s_{53} \} 
\end{align*}
on the other hand.

\blankline

\begin{lem}
  Let $L = (S, \Act, {\rightarrow}, \pi_0)$ be a deterministic LTS\@.
  \begin{enumerate}[(a)]
  \item \label{es-fact-1} If $| \Act | = 1$ then $\ES(L)$ consists of
    all cycles in~$L$.
  \item \label{es-fact-2} Every $s\in S$ has a path to an end
    structure of~$L$.
\end{enumerate}
\end{lem}

\begin{proof}
  \ref{es-fact-1} 
  Since an end structure~$S'$ is closed under
  transitions, $S'$~is a lasso. Because $S'$ is minimal and non-empty,
  it follows that $S'$ is a cycle.

  \ref{es-fact-2}
  Let~$U = \lc t \in S \mid s \pijlster{w} t ,\, w \in \calA^\ast \rc$
  be the set of states reachable from state~$s$. Then the set~$U$ is
  closed under transitions. The minimal non-empty subset
  $U' \subseteq U$ which is still closed under transitions is an end
  structure of~$L$ and can be reached by~$s$.
\end{proof}

\noindent
Next we enhance the notion of a partition refinement
algorithm. Now, an oracle can be consulted for the states in
the end structures. In this approach, the initial partition is
replaced by a partition in which all bisimilarity equivalence classes
of states in end structures are split off from the original blocks.

\begin{defi}
  A partition refinement algorithm with end structure oracle yields
  for an LTS $L = (S, \Act, {\rightarrow}, \pi_0)$ a valid refinement
  sequence~$\Pi = ( \pi'_0, \pi_1, \ldots, \pi_n)$ where $\pi'_0$ is
  the end structure partition of~$L$. The partition~$\pi'_0$ is called
  the updated initial partition of~$L$.
\end{defi}

\noindent
As Roberts' algorithm witnesses, in the case of a singleton action set
the availability of an end structure oracle results in an algorithm with
linear asymptotic performance. In the remainder of this section we confirm
that in the case of more action labels the end structure does not
help. The next lemma states that the amount of work required for the
bisplitter~$\calB_k$ by a partition refinement algorithm enhanced with
an oracle, dealing with end structures, is at least the amount of work
needed by a partition refinement algorithm without oracle for the
bisplitter~$\calB_{k{-}2}$.

\begin{lem}
  \label{lemma:oracle_projection}
  For the bisplitter $\calB_k = (S, \Act, {\rightarrow}, \pi_0)$,
  for some~$k > 2$, let $\pi_0'$ be the updated initial
  partition. Then every valid refinement sequence
  $\Pi = ( \pi_0', \pi_2, \ldots, \pi_n )$ for the updated bisplitter
  $\calB'_k = (S, \Act, {\rightarrow}, \pi'_0)$ satisfies
  $\IRC(\Pi) \geqslant \IRC(\calB_{k{-}2})$.
\end{lem}

\begin{proof}
  Observe that there are only two end structures in~$\calB_k$, viz.\
  the singletons of the two states $\bit0^k$
  and~$\bit1 \bit0^{k{-}1}$. Since all other states can reach
  $\bit0^k$ or~$\bit1 \bit0^{k{-}1}$, these states are not in an end
  structure:
  Choose $\sigma \in \Bit^k$,
  $\sigma \neq \bit0^k, \bit1 \bit0^{k{-}1}$.
  Then $\sigma$ is of the form $b \mkern1.5mu \bit0^j \bit1 \theta$ for
  some $b \in \Bit$, $j \geqslant 0$ and~$\theta \in \Bit^\ast$. For
  $j = 0$ we have
  $\sigma \pijl{a_1} \overline{b} \mkern1.5mu \bit0^{k{-}1}$ which is
  either $\bit0^k$ or $\bit1 \bit0^{k{-}1}$; for $j > 0$ we have
  $\sigma \pijl{a_{j{+}1}} b \mkern1mu \bit0^{j-1} \bit1
  \bit0^{k{-}(j{+}1)}$ while
  $b \mkern1.5mu \bit0^{j-1} \bit1 \bit0^{k{-}(j{+}1)}$ reaches
  $\bit0^k$ or $\bit1 \bit0^{k{-}1}$ by induction.  

  By Lemma~\ref{thm:facts}, every state~$\sigma \in \Bit^k$
  of~$\calB_k$ has its own bisimulation equivalence
  class~$\{ \sigma \}$. It follows that the updated initial
  partition~$\pi'_0$ consists of the blocks $\{ \bit{0}^k \}$,
  $\{ \bit{10}^{k{-}1} \}$, $B_\bit{0}^k = B_\bit{0} \backslash \{ \bit{0}^k \}$,
  and $B_\bit{1}^k = B_\bit{1} \backslash \{ \bit{10}^{k{-}1} \}$.
%
%
  Now, assume $\Pi = ( \pi_0', \pi_1, \ldots \pi_n )$ is a valid
  refinement sequence for~$\calB_{k}$. We show that
  $\IRC(\Pi) \geqslant \IRC(\calB_{k{-}2})$ by constructing a valid
  refinement sequence~$\Pi'$ for~$\calB_{k{-}2}$ satisfying
  $\IRC(\Pi) \geqslant \IRC(\Pi') \geqslant \IRC(\calB_{k{-}2})$.

  To construct $\Pi'$ from~$\Pi$, we use the partial projection
  function~$p : \Bit^k \rightharpoonup \Bit^{k{-}2}$ that
  removes the prefix~$\bit{11}$ from a bitstring and is undefined
  if~$\bit{11}$ is not a prefix.  That means
  $p(\bit{11}\sigma) = \sigma$ for all $\sigma\in \Bit^{k{-}2}$ and
  $p(\sigma')$ is undefined for
    $\sigma'\not\in \bit{11} \Bit^{k{-}2}$. 
  A partition~$\pi$ of~$\Bit^k$ is projected to a partition
  of~$\Bit^{k{-}2}$ by projecting all the blocks of~$\pi$
  and ignoring empty results, thus
  \begin{displaymath}
    p(\pi) = \lc p[B] \mid B\in \pi \rc \setminus \emptyset \mkern1mu .
  \end{displaymath}
  In particular,
  $p(\pi_0') = \{ \mkern1mu \lc \sigma \mid \sigma \in \Bit^{k{-}2}
  \rc \mkern1mu \}$, i.e.\ the unit partition
  of~$\Bit_{k{-}2}$ consisting of the prefix
  block~$B_\varepsilon$ only.
  Second, we remove repeated partitions from the sequence
  $( \mkern1mu p(\pi_0'), p(\pi_1), \ldots , p(\pi_n) )$ 
  to obtain a subsequence~$\Pi'$,
  say $\Pi' = (\varrho_0, \varrho_1, \ldots, \varrho_\ell)$.
  Thus, for some order preserving surjection
  $q : \{ 1,...,n \} \to \{ 1,...,\ell \}$ it holds that
  $p(\pi_i) = p(\pi_{i'})$ iff $q(i) = q(i')$, and
  $\varrho_j = p(\pi_i)$ if $q(i) = j$ for
    $1 \leqslant i \leqslant n$, $1 \leqslant j \leqslant \ell$.

  We have $\varrho_0 = p(\pi'_0) = \{ B_\varepsilon \}$. Next we claim
  that
  $\varrho_1 = \pi^{k{-}2}_0 = \{\bit{0}\Bit^{k-3},
  \bit{1}\Bit^{k-3}\}$ the initial partition of~$\calB_{k{-}2}$,
  containing the prefix blocks of $\bit0$ and~$\bit1$
  of~$\Bit^{k{-}2}$: Suppose to the contrary that
  $b \mkern1mu \theta, b \mkern1mu \theta' \in \Bit^{k{-}2}$, for a
  bit~$b \in \Bit$ and strings $\theta, \theta' \in \Bit^{k{-}3}$, are
  two different states which are not in the same block of~$\varrho_1$.
  Let~$i$, $0 \leqslant i < n$ be such that $p(\pi_i) = \varrho_0$ and
  $p(\pi_{i{+}1}) = \varrho_1$. Then
  $\bit{11} \mkern1mu b \mkern1mu \theta$
  and~$\bit{11} \mkern1mu b \mkern1mu \theta'$ have been separated
  when refining~$\pi_i$ into~$\pi_{i{+}1}$. But no action~$a_j$
  witnesses such a split:
  (i)~$\calB_k(\bit{11} \mkern1mu b \mkern1mu \theta, a_1) =
  \calB_k(\bit{11} \mkern1mu b \mkern1mu \theta', a_1)$ as both
  equal~$\bit0^k$;
  (ii)~$\calB_k(\bit{110} \mkern1mu \theta, a_2) = \bit{110} \mkern1mu
  \theta \in B_\bit{1}^k$, and
  $\calB_k(\bit{110} \mkern1mu \theta', a_2) = \bit{110} \mkern1mu
  \theta' \in B_{\bit1}^k$;
  (iii)~$\calB_k(\bit{111} \mkern1mu \theta, a_2) = \calB_k(\bit{111}
  \mkern1mu \theta', a_2) = 10^{k{-}1} \in B_{\bit1}^k$;
  (iv)~for $j > 2$ it holds that $\calB_k(\bit{11} \mkern1mu b
  \mkern1mu \theta , a_j), \calB_k(\bit{11} \mkern1mu b \mkern1mu
  \theta', a_j) \in B_1^k$. Since $\varrho_1 \neq
  \varrho_0$,
  $\varrho_1$ has at least two blocks. Hence, these must be
  $\bit0\Bit^{k-3}$ and~$\bit{1}\Bit^{k-3}$.  Thus $\varrho_1 = \{
  \bit0\Bit^{k-3}, \bit{1}\Bit^{k-3} \}$ as claimed.
	
  Next we prove that every refinement of~$\varrho_i$
  into~$\varrho_{i{+}1}$ of~$\Pi'$, for~$i$, $1 \leqslant i < \ell$,
  is valid for~$\calB_{k{-}2}$. We first observe that, for all
  $\sigma, \sigma' \in \Bit^{k{-}2}$, $a_j \in \Act$, it holds that
  $\calB_{k{-}2}(\sigma,a_j) = \sigma'$ iff
  $\calB_k(\bit{11} \mkern1mu \sigma, a_{j{+}2}) = \bit{11} \mkern1mu
  \sigma'$. This is a direct consequence of the definition of the
  transition functions of $\calB_{k{-}2}$ and~$\calB_k$. From this we
  obtain
  \begin{equation}  
    \sigma =_{\varrho_i} \sigma' \iff \bit{11} \mkern1mu
    \sigma =_{\pi_h} \bit{11} \mkern1mu \sigma' 
      \label{eq-rhoi-pih}
  \end{equation}
  provided $\varrho_i = p(\pi_h)$, for~$0 \leqslant i \leqslant
  \ell$ and a suitable choice of
  $h$, via the definition of the projection
  function~$p$. Now, consider the subsequent partitions
  $\varrho_i$ and~$\varrho_{i+1}$ in~$\Pi'$, $1 \leqslant i \leqslant
  \ell$. Now, let~$h$, $0 \leqslant h <
  n$, be such that $\varrho_i = p(\pi_h)$ and $\varrho_{i{+}1} =
  p(\pi_{h{+}1})$. Clearly,
  $\varrho_{i{+}1}$~is a refinement of~$\varrho_i$; if for~$B \in
  \pi_{h{+}1}$ we have $B = \bigcup_{\alpha {\in} I} \:
  B_\alpha$ with $B_\alpha \in \pi_h$ for~$\alpha \in
  I$, then for~$p[B] \in \varrho_{i{+}1}$ we have $p[B] =
  \bigcup_{\alpha {\in} I} \: p[B_\alpha]$ with $p[B_\alpha] \in
  \varrho_i$ for~$\alpha \in
  I$. The validity of the refinement
  of~$\varrho_i$
  into~$\varrho_{i{+}1}$ is justified by the validity
  of~$\pi_{h}$ into~$\pi_{h{+}1}$. If $\sigma =_{\varrho_i}
  \sigma'$ and $\sigma \neq_{\varrho_{i{+}1}} \sigma'$ for $\sigma,
  \sigma' \in \Bit^{k{-}2}$, then $\sigma, \sigma' \in
  \bit0\Bit^{k-3}$ or $\sigma, \sigma' \in
  \bit1\Bit^{k-3}$ since $\varrho_i$ is a refinement
  of~$\varrho_0$. Moreover, $\bit{11} \mkern1mu \sigma =_{\pi_h}
  \bit{11} \mkern1mu \sigma'$ and $\bit{11} \mkern1mu \sigma
  \neq_{\pi_{h{+}1}}11\sigma'$ by~(\ref{eq-rhoi-pih}). Hence, by
  validity, $\calB_k(\bit{11} \mkern1mu \sigma,a_j) \neq_{\pi_h}
  \calB_k(\bit{11} \mkern1mu \sigma',a_j)$ for some $a_j \in
    \calA$.  Clearly $j \neq 1$, since $(\bit{11}\sigma)[2] =
    (\bit{11}\sigma')[2] = \bit{1}$.  Also,~$j \neq
  2$, since $\sigma[1] = \sigma'[1]$ we have $(\bit{11} \mkern1mu
  \sigma)[3] = (\bit{11} \mkern1mu
  \sigma')[3]$. Therefore, $\calB_{k{-}2}(\sigma,a_{j{-}2})
  \neq_{\varrho_i}
  \calB_{k{-}2}(\sigma',a_{j{-}2})$, showing the refinement
  of~$\varrho_i$ into~$\varrho_{i{+}1}$ to be valid.
  
  Finally, since every block in~$\pi_n$ is a singleton, this is also
  the case for~$\varrho_\ell$. Thus, $\varrho_\ell$ is indeed the
  coarsest stable partition for~$\calB_{k{-}2}$ as required for~$\Pi'$
  to be a valid refinement sequence for~$\calB_{k{-}2}$. Every
  refinement of~$\varrho_i$ into~$\varrho_{i{+}1}$ of~$\Pi'$ is
  projected from a refinement of some~$\pi_h$ into~$\pi_{h{+}1}$
  of~$\Pi$ as argued above.
  Therefore, since $p(\pi_h) = \varrho_i$
  and $p(\pi_{h{+}1}) = \varrho_{i{+}1}$, we have
  $\IRC(\pi_h,
  \pi_{h{+}1}) \geqslant \IRC(\varrho_i, \varrho_{i{+}1})$, and hence
  $\IRC(\Pi) = \sum_{h{=}1}^{n} \, \IRC(\pi_{h{-}1},\pi_h) \geqslant
  \sum_{i{=}1}^\ell \, \IRC(\varrho_{i{-}1},\varrho_i) = \IRC(\Pi')$.
  Since, by definition, $\IRC(\calB_{k{-}2})$ is the minimum over all
  valid refinement sequences for~$\calB_{k{-}2}$ it holds that
  $\IRC(\Pi') 
  \geqslant \IRC(\calB_{k{-}2})$. Therefore, $\IRC(\Pi) \geqslant
  \IRC(\Pi') \geqslant \IRC(\calB_{k{-}2})$ as was to be shown.
\end{proof}

\noindent
Next we combine the above lemma with the lowerbound provided by
Theorem~\ref{thm:inherent_complexity} in order to prove the main
result of this section.

\begin{thm}
  Any partition refinement algorithm with an end structure oracle to
  decide bisimilarity for a deterministic LTS is $\Omega(n\log n)$.
\end{thm}

\begin{proof}
  \label{thm-part-ref-with-oracle}
  Let $\calB'_k$ be the updated bisplitter with the initial
  partition~$\pi'_0$ containing $\{ \bit0^k \}$,
  $B_{\bit0} \backslash \{ \bit0^k \}$, $\{ \bit{10}^{k{-}1} \}$, and
  $B_{\bit1} \backslash \{ \bit{10}^{k{-}1} \}$ as given by the oracle
  for end structures rather than the partition~$\pi_0$ containing
  $B_{\bit0}$ and~$B_{\bit1}$.  By
  Lemma~\ref{lemma:oracle_projection} we have, for $k > 2$, that
  $\IRC(\calB'_k) \geqslant \IRC(\calB_{k{-}2})$. By
  Theorem~\ref{thm:inherent_complexity} we know that
  $\IRC(\calB_{k{-}2}) \geqslant \frac{1}{2} \tilde{n} \log \frac12
  \tilde{n}$ for $\tilde{n}= 2^{k{-}2}$, the number of states
  of~$\calB_{k{-}2}$. It holds that
  $\tilde{n} = \frac{2^{k{-}2}}{2^{k}} \mkern1mu n = \frac{1}{4} n$. So
  $\IRC(\calB'_k) \geqslant \frac{1}{8} n \log \frac{1}{8} n$ from
  which we conclude that deciding bisimilarity for~$\calB_k$ with the
  help of an oracle for the end structures is $\Omega(n \log n)$.
\end{proof}


\section{$\calC_k$ is $\Omega((m+n)\log n)$ for partition refinement}
\label{sec:two-action-action-set}

We modify the bisplitter~$\calB_k$, that has an action alphabet of
$k{-}1$~actions, to obtain a deterministic LTS with two actions
only. The resulting LTS~$\calC_k$ has the action alphabet
$\{ a,b \}$, for each~$k > 1$, and is referred to as the $k$-th
\emph{layered bisplitter}. We use~$\calC_k$ to obtain an
$\Omega((n+m) \log n)$ lowerbound for deciding bisimilarity for LTSs
with only two actions, where $n$~is the number of states and $m$~is the
number of transitions.

In order to establish the lowerbound we adapt the construction
of~$\calB_k$ at two places. We introduce for each
$\sigma \in \Bit^k$, a stake of $2^k$~states.
Moreover, to each stake we add a tree gadget. These gadgets have height
$\lceil {\log (\frac{k-1}{2})} \rceil$ to accommodate
$\lceil (k{-}1)/2 \rceil$~leaves in order to encode the action alphabet~$\Act_k$
	of~$\calB_k$ with $k{-}1$~actions.

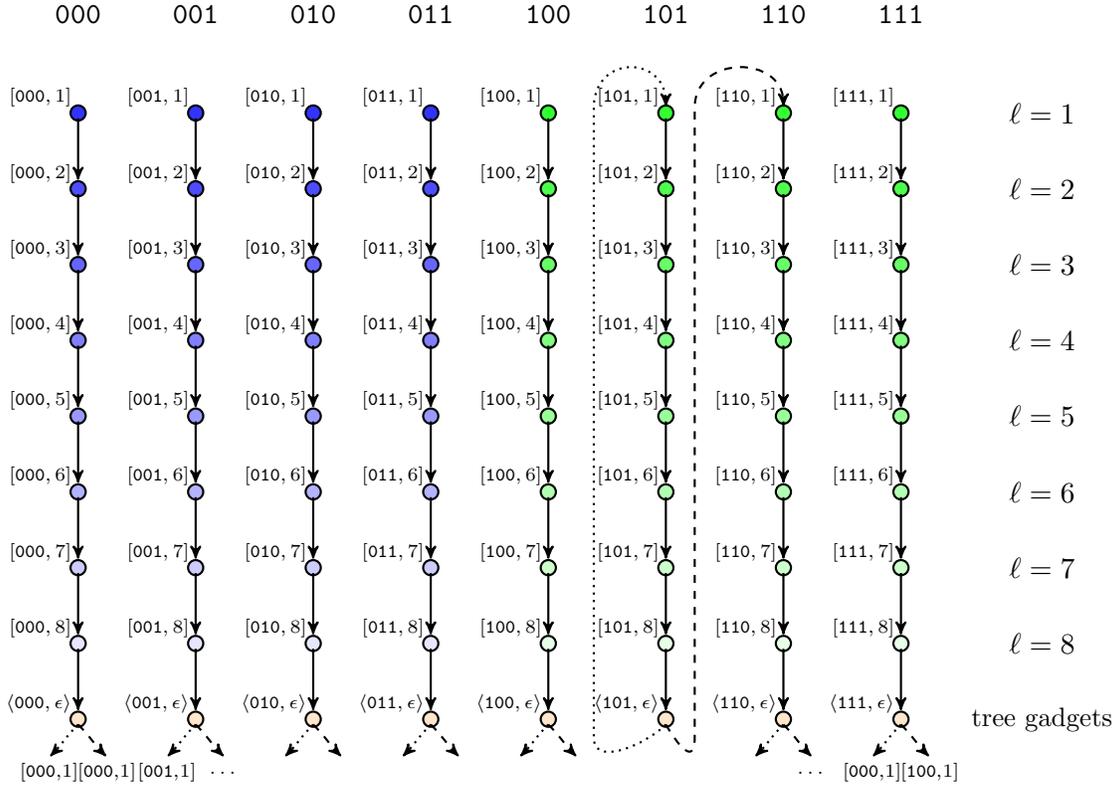
\begin{figure}[t]
  \begin{center}
    \resizebox{\textwidth}{!}{

\begin{tikzpicture}[-]
  \def\colwidth{1.55}
  \foreach \y / \label in {0/ \bit{000}, 1/\bit{001}, 2/\bit{010},
    3/\bit{011}} { 
    \foreach \x  in {1, ..., 8} {
      \pgfmathsetmacro\result{(9 - \x)}
      \pgfmathsetmacro\adjusty{(\y * \colwidth)}
      \draw[thick, fill=blue!\x0, font=\tiny] (\adjusty,\x)
      circle(0.1cm) node[above left]
      {$[{\label},\pgfmathprintnumber{\result}]$}; 
    }
  }
	
  \foreach \y / \label in {4/\bit{100}, 5/\bit{101}, 6/\bit{110}, 7/\bit{111}} {
    \foreach \x in {8,...,1} {
      \pgfmathsetmacro\result{(9 - \x)}
      \pgfmathsetmacro\adjusty{(\y * \colwidth)}
      \draw[thick,fill=green!\x0, font=\tiny] (\adjusty,\x)
      circle(0.1cm) node[above left]
      {$[\label,\pgfmathprintnumber{\result}]$}; 
    }
  }

  \foreach \y / \label in {
    0/\bit{000},
    1/\bit{001},
    2/\bit{010},
    3/\bit{011},
    4/\bit{100},
    5/\bit{101},
    6/\bit{110},
    7/\bit{111}} {
    \pgfmathsetmacro\x{(0)};
    \pgfmathsetmacro\adjusty{(\y * \colwidth)};
    \draw[thick,fill=orange!20, font=\tiny] (\adjusty,\x)
    circle(0.1cm) node[above left]
    {$\langle{\label},\epsilon\rangle$}; 
  }
	
  \def\indexheight{9.3}
  \foreach \y / \label in {
    0/\bit{000},
    1/\bit{001},
    2/\bit{010},
    3/\bit{011},
    4/\bit{100},
    5/\bit{101},
    6/\bit{110},
    7/\bit{111}} {
    \pgfmathsetmacro\adjusty{(\y * \colwidth)};
    \draw (\adjusty,\indexheight) node {$\label$};
  }
		
		
  \foreach \col in {0,...,7} {
    \foreach \row in {7,...,0} {
      \pgfmathsetmacro\adjusty{(\col * \colwidth)}
      \draw[thick, ->](\adjusty,\row.93) to [out=-90,in=90,
      distance=0cm ] (\adjusty,\row.1); 
    }
  }

  \pgfmathsetmacro\0{(0 * \colwidth)};
  \pgfmathsetmacro\1{(1 * \colwidth)};
  \pgfmathsetmacro\2{(2 * \colwidth)};
  \pgfmathsetmacro\3{(3 * \colwidth)};
  \pgfmathsetmacro\3{(3 * \colwidth)};
  \pgfmathsetmacro\4{(4 * \colwidth)};
  \pgfmathsetmacro\5{(5 * \colwidth)};
  \pgfmathsetmacro\6{(6 * \colwidth)};
  \pgfmathsetmacro\7{(7 * \colwidth)};
  \pgfmathsetmacro\diff{(0.38)};
  \pgfmathsetmacro\leafheight{(-0.5)};
  \pgfmathsetmacro\leafrangle{()};
  \pgfmathsetmacro\leaflangle{()};
		
  \pgfmathsetmacro\kwidth{(8.2 * \colwidth)}
  \foreach \k in {8,...,1} {
    \pgfmathsetmacro\result{9 - \k};
    \draw [-] (\kwidth,\k) node {$\ell=\pgfmathprintnumber{\result}$};
  }	
  \draw [-] (\kwidth,0) node {{\small tree gadgets}};

  \draw[thick, dotted, ->](\0,-0.07) to [out=225,in=45]
  ({\0-\diff},\leafheight); 
  \draw[thick, dashed, ->](\0,-0.07) to [out=315,in=135]
  ({\0+\diff},\leafheight);		 
  \draw (\0-\diff,\leafheight) node[below,font=\tiny] {$[\bit{000},\!1]$};
  \draw (\0+\diff,\leafheight) node[below,font=\tiny] {$[\bit{000},\!1]$};
		
  \draw[thick, dotted, ->](\1,-0.07) to [out=225,in=45]
  ({\1-\diff},\leafheight); 
  \draw[thick, dashed, ->](\1,-0.07) to [out=315,in=135]
  ({\1+\diff},\leafheight); 
  \draw (\1-\diff,\leafheight) node[below,font=\tiny] {$[\bit{001},\!1]$};
  \draw (\1+\diff,\leafheight) node[below=3pt,font=\tiny] {$\dots$};

  \draw[thick, dotted, ->](\2,-0.07) to [out=225,in=45]
  ({\2-\diff},\leafheight); 
  \draw[thick, dashed, ->](\2,-0.07) to [out=315,in=135]
  ({\2+\diff},\leafheight); 
		
  \draw[thick, dotted, ->](\3,-0.07) to [out=225,in=45]
  ({\3-\diff},\leafheight); 
  \draw[thick, dashed, ->](\3,-0.07) to [out=315,in=135]
  ({\3+\diff},\leafheight); 
		
  \draw[thick, dotted, ->](\4,-0.07) to [out=225,in=45]
  ({\4-\diff},\leafheight); 
  \draw[thick, dashed, ->](\4,-0.07) to [out=315,in=135]
  ({\4+\diff},\leafheight); 
		
  \draw[thick, dotted, ->](\5,-0.07) to [out=225,in=270]
  ({\5-2.5*\diff},\leafheight+0.25) to [] (\5-2.5*\diff,8.07) to
  [out=90, in=90,distance=20] (\5,8.07); 
		
  \draw[thick, dashed, ->](\5,-0.07) to [out=300,in=270,distance=10]
  ({\5+\diff},{\leafheight+0.25}) to [out=90, in=270] (\5+\diff,8.07)
  to[out=90, in=90, distance=20] (\6,8.07);

		
  \draw[thick, dotted, ->](\6,-0.07) to [out=225,in=45]
  ({\6-\diff},\leafheight); 
  \draw[thick, dashed, ->](\6,-0.07) to [out=315,in=135]
  ({\6+\diff},\leafheight); 
  \draw (\6+\diff,\leafheight) node [below=3pt,font=\tiny] {$\dots$};
		
  \draw[thick, dotted, ->](\7,-0.07) to [out=225,in=45]
  ({\7-\diff},\leafheight); 
  \draw[thick, dashed, ->](\7,-0.07) to [out=315,in=135]
  ({\7+\diff},\leafheight); 
  \draw (\7-\diff,\leafheight) node[below,font=\tiny]
  {$[\bit{000},\!1]$};
  \draw (\7+\diff,\leafheight) node[below,font=\tiny] {$[\bit{100},\!1]$};
		
%
%
%
%
%
%
%

              \end{tikzpicture}
    }
  \end{center}
  \caption{The partial layered bisplitter
    $\calC_3$ with tree gadgets, the colours represent the initial partition.} 
  \label{fig:bisplitter3_8}
\end{figure}

\begin{defi}
  \label{def:bisplitter-twolabels}
  Let $k > 1$, $\calB_k$ be the $k$-th bisplitter, and
  $\ofA = \{a,b\}$ be a two-element action set. The deterministic
  LTS~$\calC_k = ( S^\calC_k , \ofA, {\rightarrow_{\calC}},
  \pi_0^\calC)$, over the action set~$\ofA \mkern1mu$,
  \begin{enumerate}[(a)]
  \item has the set of
  states~$S^\calC_k$ defined as \smallskip \\ \qquad
    \begin{math}
      \begin{array}[t]{rcl}
        S^\calC_k
        & =
        & \lc \bpair{\sigma}{\ell} \in \Bit^k \times
          \Nat \mid 1 \leqslant \ell \leqslant 2^k \rc \cup {}
          \smallskip \\
        & 
        & \lc \apair{\sigma}{w} \in \Bit^k \times \ofA^\ast \mid 0
          \leqslant |w| \leqslant \lceil \log (\frac{k-1}{2}) \rceil \rc 
          \mkern1mu ,
      \end{array}
    \end{math}
    \smallskip
  \item \label{item-trans-calC}
    has the transition relation $\rightarrow_\calC$ given by 
    \smallskip \\ \qquad
    \begin{math}
      \def\arraystretch{1.2}
      \begin{array}[t]{rclcl}
        \bpair{\sigma}{\ell}
        & \cpijl{\alpha}
        & \bpair{\sigma}{\ell+1}
        &
        & \text{for $\sigma \in \Bit^k$, $1 \leqslant \ell <
          2^k$, $\alpha \in \ofA$} \\
        \bpair{\sigma}{2^k}
        & \cpijl{\alpha}
        & \apair{\sigma}{\varepsilon}
        &
        & \text{for $\sigma \in \Bit^k$, $\alpha \in \ofA$} \\
        \apair{\sigma}{w}
        & \cpijl{\alpha}
        & \apair{\sigma}{w \mkern1mu \alpha}
        &
        & \text{for $\sigma \in \Bit^k$, $|w| < \lceil \log (\frac{k-1}{2})
        		\rceil$, $\alpha
          \in \ofA$} \\
        \apair{\sigma}{w}
        & \cpijl{\alpha}
        & \bpair{\sigma'}{1}
        &
        & \text{for $\sigma \in \Bit^k$, $|w| = \lceil \log (\frac{k-1}{2})
          \rceil$,
          $\lbl(w \mkern1mu \alpha) = j$,}\\
      & & & &\text{and, $\calB_k( \sigma,a_{j}) = \sigma'$,
      	  $\alpha \in \ofA$}, \\
      \end{array}
      \def\arraystretch{1.0}
    \end{math} \\
  \item and has the initial partition~$\pi^\calC_0= \lc
  C^{\mkern1mu \ell}_{\bit0} , C^{\mkern1mu \ell}_{\bit1}
  \mid 1 \leqslant \ell \leqslant 2^k \rc \cup \lbrace C_\varepsilon \rbrace$ defined as
      \smallskip \\ \qquad
	  \begin{math}
	  \begin{array}[t]{rclcl}
	  C^{\mkern1mu \ell}_{\bit0}
	  & =
	  & \,
	  \lc \bpair{\sigma}{\ell} \mid
	  \sigma \in B_{\bit0} \rc 
	  &
	  &\text{for $1\leqslant l\leqslant 2^k$}\\
	  	  C^{\mkern1mu \ell}_{\bit1}
	  & =
	  & \,
	  \lc \bpair{\sigma}{\ell} \mid
	  \sigma \in B_{\bit1} \rc 
	  &
	  &\text{for $1\leqslant l\leqslant 2^k$}\\
	  C_\varepsilon 
	  & =
	  & \lc \apair{\sigma}{w} \in S^\calC_k \mid \sigma \in \Bit^k ,\, w
	  \in \ofA^\ast \rc.
	  &
	  &
	  \end{array}
	  \end{math}
  \end{enumerate}
  The auxiliary labelling function
  $\lbl : \ofA^{\leqslant \lceil\log (k{-}1)\rceil} \to \Nat$, used in
  item~\ref{item-trans-calC} is defined by
  $\lbl(w) = \min \{ \mathit{bin}(w) {+}1, k{-}1 \}$. Here
  $\mathit{bin} : \ofA^* \to \Nat$ is the binary evaluation function
  defined by $\mathit{bin}(\varepsilon) = 0$,
  $\mathit{bin}(w \mkern1mu a) = 2 * \mathit{bin}(w)$, and
  $\mathit{bin}(w \mkern1mu b) = 2 * \mathit{bin}(w) {+}1$.
\end{defi}

\noindent
We see that with each string $\sigma \in \Bit^k$ we associate
in~$\calC_k$ as many as $2^k$~stake states
$\bpair{\sigma}{1}, \ldots, \bpair{\sigma}{2^k}$, one for each
level~$\ell$, $1 \leqslant \ell \leqslant 2^k$. The stake states are
traversed from the top~$\bpair{\sigma}{1}$ to
bottom~$\bpair{\sigma}{2^k}$ for each string~$\sigma$ over~$\ofA$ of
length~$2^k$. The tree gadget, with states $\bpair{\sigma}{w}$ for bit
sequences~$\sigma$ and strings~$w$ over~$\ofA$, consists of a complete
binary tree of height~$\lceil \log(\frac{k-1}{2}) \rceil$ that hence
has $\lceil(k-1)/2\rceil$~leaves. Traversal down
the tree takes a left child on action~$a$ from~$\ofA$, a right child
on action~$b$ from~$A$. Together with the two actions of~$\ofA$,
$k{-}1$~source-label pairs can be encoded, connecting the stake on top
of the tree gadgets $k{-}1$ times with other stakes. To simulate a
transition $\sigma \pijl{a_j} \sigma'$ of~$\calB_k$ in~$\calC_k$ from
a leaf of a tree gadget of~$\sigma$ to the top of the stake
of~$\sigma'$, we need to be at a leaf~$\apair{\sigma}{w}$ of the tree
gadget of~$\sigma$ such that the combined string $w \mkern1mu \alpha$
for~$\alpha \in \ofA$ is the binary encoding according to~$\lbl$ of
the index~$j$. An $\alpha$-transition thus leads from the
source~$\apair{\sigma}{w}$ to the target~$\bpair{\sigma'}{1}$ if
$\sigma \pijl{a_j} \sigma'$ in~$\calB_k$ and $w \mkern1mu \alpha$
corresponds to~$j$.
The
partition
\begin{displaymath}
  \pi^\calC_0 = \lc
  C^{\mkern1mu \ell}_{\bit0} , C^{\mkern1mu \ell}_{\bit1}
  \mid 1 \leqslant \ell \leqslant 2^k \rc
  \cup \lbrace C_\varepsilon \rbrace
\end{displaymath}
distinguishes, for each level~$\ell$, the states at level~$\ell$ of
the stakes of strings starting with~$\bit0$ in~$C^\ell_{\bit0}$, the
states of the stakes at level~$\ell$ of strings starting with~$\bit1$
in~$C^\ell_{\bit1}$, and the states of the tree gadgets collected
in~$C_\varepsilon$.

\blankline

\noindent
Figure~\ref{fig:bisplitter3_8} depicts the layered
$3$-splitter~$\calC_3$. Because also $\calB_3$ has an action set of
size~$2$ the tree gadgets only consist of the root node of the form
$\apair{\sigma}{\varepsilon}$. In Figure~\ref{fig:bisplitter12} we see
that for bisplitter~$\calB_3$ we have $\bit{101} \pijl{a_1} \bit{101}$
and $\bit{101} \pijl{a_2} \bit{110}$. In
Figure~\ref{fig:bisplitter3_8} we have transitions
$\apair{\bit{101}}{\varepsilon} \pijl{a} \bpair{\bit{101}}{1}$ and
$\apair{\bit{101}}{\varepsilon} \pijl{b} \bpair{\bit{110}}{1}$ (dotted
and dashed, respectively).
Colouring of nodes is used to represent the initial
partition~$\pi^\calC_0$ that contains $17$~blocks:
for each level~$\ell$, $1\leqslant \ell \leqslant 2^3$, 
$\pi^\calC_0$ contains a block holding the four states of the stakes in~$C^\ell_{\bit0}$ on the
left and a block with the four stake states in~$C^\ell_{\bit1}$ on the right, and 
lastly one block consisting of the eight tree states in~$C_{\varepsilon}$ at
the bottom of the picture.

The $6$-th bisplitter~$\calB_6$ has five actions, $a_1$ to~$a_5$. A
tree gadget for the layered bisplitter~$\calC_6$ with corresponding
outgoing transitions is drawn in Figure~\ref{fig:tree_c5}. The tree
has height
$\lceil\log ((6-1)/2)\rceil = \lceil\log \frac{5}{2}\rceil = 2$, hence
it has $2^2 = 4$~leaves. Since each leaf has two outgoing transitions,
one labelled~$a$ and one labelled~$b$, the two leftmost leaves
$\apair{\sigma}{aa}$ and~$\apair{\sigma}{ab}$ are used with the two
labels $a$ and~$b$ to simulate transitions for $a_1$ up to~$a_4$, the
two rightmost leaves $\apair{\sigma}{ba}$ and~$\apair{\sigma}{bb}$
have together four transitions all simulating the $a_5$-transition
of~$\sigma$.

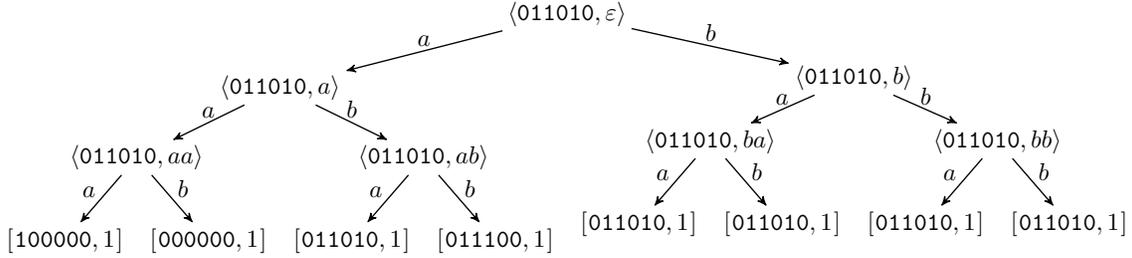
\begin{figure}[t]
  \centering
  \resizebox{\textwidth}{!}{
  \begin{forest}
    [{$\apair{\bit{011010}}{\varepsilon}$}, name=node 1
    [{$\apair{\bit{011010}}{a}$}, l=11mm, edge label={node[midway,
      above]{$a$}} 		 	
    [{$\apair{\bit{011010}}{aa}$}, l=11mm, edge label={node[midway, above]{$a$}}
    [{$\bpair{\bit{100000}}{1}$}, l=13mm, edge label={node[midway,
      above left, yshift=-0.5mm]{$a$}}] 
    [{$\bpair{\bit{000000}}{1}$}, l=13mm, edge label={node[midway,
      above right, yshift=-0.5mm]{$b$}}] 
    ]
    [{$\apair{\bit{011010}}{ab}$}, l=11mm, edge label={node[midway, above]{$b$}}
    [{$\bpair{\bit{011010}}{1}$}, l=13mm, edge label={node[midway,
      above left, yshift=-0.5mm]{$a$}}] 
    [{$\bpair{\bit{011100}}{1}$}, l=13mm, edge label={node[midway,
      above right, yshift=-0.5mm]{$b$}}] 
    ]
    ]
    [{$\apair{\bit{011010}}{b}$}, edge label={node[midway, above]{$b$}} 
    [{$\apair{\bit{011010}}{ba}$}, edge label={node[midway, above]{$a$}}
    [{$\bpair{\bit{011010}}{1}$}, l=13mm, edge label={node[midway,
      above left, yshift=-0.5mm]{$a$}}] 
    [{$\bpair{\bit{011010}}{1}$}, l=13mm, edge label={node[midway,
      above right, yshift=-0.5mm]{$b$}}] 
    ]
    [{$\apair{\bit{011010}}{bb}$}, edge label={node[midway, above]{$b$}}
    [{$\bpair{\bit{011010}}{1}$},  l=13mm, edge label={node[midway,
      yshift=-0.5mm, above left]{$a$}}] 
    [{$\bpair{\bit{011010}}{1}$}, l=13mm, edge label={node[midway, above right,
      yshift=-0.5mm]{$b$}}]
    ]
    ]
    ]
  \end{forest}
  } 
  \caption{The example of the outgoing tree for $\calC_6$ from the root
    $\bpair{\bit{011010}}{\varepsilon}\in S^{\calC}_6$\!.}
  \label{fig:tree_c5}
\end{figure}

\blankline

\noindent
The next lemma introduces three facts for the layered
bisplitter~$\calC_k$ that we need in the sequel. The first states that
if two states at different stakes, but at the same level, are
separated during partition refinement, then all corresponding states
at lower levels are separated as well.
The second fact helps to transfer witnessing transitions in~$\calB_k$
to the setting of~$\calC_k$. A transition $\sigma \pijl{a_j} \sigma'$
of~$\calB_k$ is reflected by a path from $\bpair{\sigma}{2^k}$ through
the tree gadget of~$\sigma$ from root to leaf and then to the top
state~$\bpair{\sigma'}{1}$ of the stake of~$\sigma'$. The
word~$w \mkern1mu \alpha$ encountered going down and out the tree
gadget corresponds to the action~$a_j$ according to the
$\lbl$-function.
Lastly, it is shown that
no two pairs of different states within the stakes are bisimilar.

\begin{lem}
  \label{lem:heights}
  Let $\Pi$ be a valid refinement sequence for~$\calC_k$ and $\pi$ a
  partition in~$\Pi$.
  \begin{enumerate}[(a)]
  \item \label{fact-4-1} If two states
    $\bpair{\sigma}{\ell}, \bpair{\sigma'}{\ell} \in S^\calC_k$,
    for~$1 \leqslant \ell \leqslant 2^k$, are in a different block
    of~$\pi$, then all pairs
    $\bpair{\sigma}{m}, \bpair{\sigma'}{m} \in S$, for all levels~$m$,
    $\ell \leqslant m \leqslant 2^k$, are in different blocks
    of~$\pi$.

    \smallskip
    
  \item \label{fact-4-2} If $\bpair{\sigma_1}{2^k}$
    and~$\bpair{\sigma_2}{2^k}$ are split for~$\pi$, then
    there are $w \in \ofA^\ast$, $\alpha \in \ofA$, and $\sigma'_1,
    \sigma'_2 \in \Bit^k$ such that
    \begin{displaymath}
      \bpair{\sigma_1}{2^k} \pijlsterc{w} \apair{\sigma_1}{w}
      \cpijl{\alpha} \bpair{\sigma_1'}{1}
      \quad \text{and} \quad
      \bpair{\sigma_2}{2^k} \pijlsterc{w} \apair{\sigma_2}{w}
      \cpijl{\alpha} \bpair{\sigma_2'}{1}
    \end{displaymath}
    with $\bpair{\sigma_1'}{1}$ and~$\bpair{\sigma_2'}{1}$ in different
    blocks of~$\pi$.

    \smallskip
    
  \item \label{fact-4-3} If $\pi$ is the last refinement in~$\Pi$, it
    contains the singletons of~$\bpair{\sigma}{\ell}$ for 
    $\sigma \in \Bit^k$ and~$1 \leqslant \ell \leqslant 2^k$.
\end{enumerate}
\end{lem}

\begin{proof}
  (a)~For a proof by contradiction, suppose the partition~$\pi$ is the
  first partition of~$\Pi$ that falsifies the statement of the lemma.
  So $\pi \neq \pi^\calC_0$, since for the initial
  partition~$\pi^\calC_0$ the statement holds. Thus,
  $\pi$ is a refinement of a partition~$\pi'$ in~$\Pi$. So, there are
  two states
  $\bpair{\sigma}{\ell}, \bpair{\sigma'}{\ell} \in S^\calC_k$ in
  different blocks of~$\pi$ while the states
  $\bpair{\sigma}{\ell{+}1}, \bpair{\sigma'}{\ell{+}1}$ are in the
  same block of~$\pi$ and hence of~$\pi'$. Since
  $\bpair{\sigma}{\ell}$ and~$\bpair{\sigma'}{\ell}$ only have
  transitions to $\bpair{\sigma}{\ell{+}1}$
  and~$\bpair{\sigma'}{\ell{+}1}$, respectively, that are in the same
  block~$\pi'$, the refinement would not have been valid. We conclude
  that no falsifying partition~$\pi$ in~$\Pi$ exists and that the
  lemma holds.

  (b)~We first prove, that for all $w\in \ofA^\ast$,
  $|w|\leqslant \lceil \log (k{-}1) \rceil - 1$, if
  $\apair{\sigma_1}{w}$ and~$\apair{\sigma_2}{w}$ are split in~$\pi$,
  then there are $v \in \ofA^\ast$ and $\alpha \in \ofA$ such that
  $\apair{\sigma_1}{w} \pijlster{v} \apair{\sigma_1}{wv} \pijl{\alpha}
  \bpair{\sigma'_1}{1}$ and
  $\apair{\sigma_2}{w} \pijlster{v} \apair{\sigma_2}{wv} \pijl{\alpha}
  \bpair{\sigma'_2}{1}$ with $\bpair{\sigma'_1}{1}$
  and~$\bpair{\sigma'_2}{1}$ in different blocks of~$\pi$. We prove
  this for all possible lengths $|w|$ by reverse induction. If~$w$ has
  maximal length, i.e.\ $|w| = \lceil \log(k{-}1) \rceil-1$ this is
  clear.  If $\apair{\sigma_1}{w}$ and~$\apair{\sigma_2}{w}$ are
  split, for $|w| < \lceil\log(k{-}1)\rceil-1$, then either
  $a$-transitions or $b$-transitions lead to split states. By the
  induction hypothesis, suitable paths exist from the targets of such
  transitions. Adding the respective transition proves the induction
  hypothesis. Since $\bpair{\sigma_1}{2^k}$
  and~$\bpair{\sigma_2}{2^k}$ can only reach
  $\apair{\sigma_1}{\varepsilon}$ and~$\apair{\sigma_2}{\varepsilon}$
  the statement follows.

  (c)~Choose~$\ell$, $1 \leqslant \ell \leqslant 2^k$ and define the
  relation $R \subseteq {S^\calB_k \times S^\calB_k}$ such that 
  $(\sigma_1, \sigma_2)\in R$ iff the stake states
  $\bpair{\sigma_1}{\ell}, \bpair{\sigma_2}{\ell}\in S^\calC_k$ are
  bisimilar for~$\calC_k$. We verify that $R$ is a bisimulation
  relation for~$\calB_k$. Note, that $R$ respects~$\pi^\calB_k$, the
  initial partition of~$\calB_k$. Now, suppose $(\sigma_1, \sigma_2)\in R$
  and
  $\sigma_1 \pijl{a_j} \sigma'_1$ for some $a_j \in \calA_k$
  and~$\sigma'_1 \in S^\calB_k$.
  By construction of~$\calC_k$ we have
  \begin{displaymath}
    \bpair{\sigma_1}{\ell} \pijlster{a^{2^k{-}\ell}}
    \bpair{\sigma_1}{2^k} \pijl{a} \apair{\sigma_1}{\varepsilon}
    \pijlster{w} \apair{\sigma_1}{w} \pijl{\alpha} \bpair{\sigma'_1}{1}
    \pijlster{a^{\ell{-}1}} \bpair{\sigma'_1}{\ell}
  \end{displaymath}
  where
  $\lbl(w \mkern1mu \alpha) = j$. Since $\bpair{\sigma_1}{\ell}$ and
  $\bpair{\sigma_2}{\ell}$ are bisimilar in~$\calC_k$, it follows that
  a corresponding path
  $\bpair{\sigma_2}{\ell} \pijlster{} \bpair{\sigma'_2}{\ell}$ exists
  in~$\calC_k$ with $\bpair{\sigma'_1}{\ell}$ and
  $\bpair{\sigma'_2}{\ell}$ bisimilar in~$\calC_k$. From this we
  derive that $\sigma_2 \pijl{a_j} \sigma'_2$ in~$\calB_k$
  and~$(\sigma'_1,\sigma'_2)\in R$. Hence, $R$~is a bisimulation relation
  for~$\calB_k$ indeed. It holds that bisimilarity of~$\calB_k$ is
  the identity relation. Thus, if two stake states $\bpair{\sigma_1}{\ell}$
  and~$\bpair{\sigma_2}{\ell}$ are bisimilar for~$\calC_k$, then
  $\sigma_1$ and~$\sigma_2$ are bisimilar for~$\calB_k$ thus
  $\sigma_1 = \sigma_2$, and therefore~$\bpair{\sigma_1}{\ell}
  = \bpair{\sigma_2}{\ell}$.
\end{proof}

\noindent
The next lemma states that the splitting of states
$\bpair{\sigma}{\ell} \in S^{\calC}$, for each level~$\ell$,
has refinement costs that are at least that of~$\calB_k$.

\begin{lem}
  \label{lemma:inflate}
  It holds that $\IRC(\calC_k) \geqslant 2^k \IRC(\calB_k)$ for all~$k > 1$.
\end{lem}
 
\begin{proof}
  Let $\Pi = ( \pi^\calC_0, \pi_1, \dots, \pi_n)$ be a valid
  refinement sequence for~$\calC_k$. We show that for each
  level~$\ell$, the sequence~$\Pi$ induces a valid refinement
  sequence~$\Pi^\ell$ for~$\calB_k$.

For each $\ell\in\Nat$, such that $1\leqslant \ell \leqslant m$, we define a
partial projection function~$p_\ell: S^\calC_k \rightharpoonup \Bit^k$. 
The mapping~$p_\ell$ maps states of shape $[\sigma,\ell]\in S^\calC_k$ of $\calC_k$
to $\sigma\in\Bit^k$ and is undefined on all other states. A block $B$
in a partition of~$\calC_k$ is mapped to the block $p_\ell[B]$ of $\Bit^k$,
by applying~$p_\ell$ on all elements, resulting in:
  \begin{displaymath}
    p_\ell[ B ] = \lc \sigma \in \Bit^k \mid
		\bpair{\sigma}{\ell} \in B \rc\mkern1mu.
  \end{displaymath}
  A partition $p_\ell(\pi)$ of $\calB_k$ is obtained by applying $p_\ell$ to a
  partition $\pi$ of $\calC_k$ and ignoring the empty blocks, i.e.
  $p_\ell (\pi) = \lc p_\ell[B] \mid B\in \pi \rc \setminus
  \emptyset$.
  The sequence~$\Pi^\ell = ( \pi^\ell_0 , \dots , \pi^\ell_m )$ is
  obtained from the sequence
  $( p_\ell( \pi^\calC_0 ), p_\ell(\pi_1), \dots, p_\ell(\pi_n))$ by
  removing possible duplicates. We verify that $\Pi^\ell$~is a valid
  refinement sequence for~$\calB_k$.

  First, we check that $\pi^\ell_{i}$ is a refinement
  of~$\pi^\ell_{i{-}1}$, for $1 \leqslant i \leqslant m$. Choose index~$i$ arbitrarily. Let the index~$h$ with
  $1 \leqslant h \leqslant n$ be such that
  $p_\ell(\pi_{h{-}1}) = \pi^\ell_{i{-}1}$ and
  $p_\ell(\pi_{h}) = \pi^\ell_i$. 
  Then we fix a block $B\in\pi^\ell_i$. Since $\pi^\ell_i = p_\ell(\pi_h)$  there is a block $B'\in \pi_{h}$ such that $B = p_\ell[B']$. Since
 	$\pi_{h}$ is a refinement of $\pi_{h{-}1}$ there is a block
 	$B''\in \pi_{h{-}1}$ such that $B'\subseteq B''$. This implies
 	that $p_\ell[B'] \subseteq p_\ell[B'']$ and since
 	$B = p_\ell[B'] \neq \emptyset$ also $p_\ell[B''] \neq \emptyset$.
 	So, we conclude that $p_\ell[B'']\in \pi^\ell_{i{-}1}$ and
 	$B\subseteq p_\ell[B'']$. Thus $\pi^\ell_{i}$ is a refinement of
 	$\pi^\ell_{i{-}1}$.

  Next, we verify that $\Pi^\ell$ is a valid refinement sequence
  for~$\calB_k$. Suppose the state $\sigma_1, \sigma_2 \in S^\calB_k$
  are split for the refinement of~$\pi^\ell_{i{-}1}$
  into~$\pi^\ell_{i}$. Then the states
  $\bpair{\sigma_1}{\ell}, \bpair{\sigma_2}{\ell} \in S^\calC_k$ are
  split for the refinement of a partition~$\pi_{h{-}1}$ into the
  partition~$\pi_{h}$ for some index~$h$, with
  $1 \leqslant h \leqslant n$. Then either (i)~$\ell = 2^k$ and
  $\bpair{\sigma_1}{\ell}$ and~$\bpair{\sigma_2}{\ell}$ have
  $\alpha$-transitions to different blocks, for
  some~$\alpha \in \ofA$, or (ii)~$\ell < 2^k$ and
  $\bpair{\sigma_1}{\ell{+}1}$ and~$\bpair{\sigma_2}{\ell{+}1}$ are in
  different blocks of~$\pi_{h{-}1}$. In the case of~(ii), it follows
  by Lemma~\ref{lem:heights} that also $\bpair{\sigma_1}{2^k}$
  and~$\bpair{\sigma_2}{2^k}$ are in different blocks
  of~$\pi_{h{-}1}$. Thus, the refinement of some~$\pi_{g{-}1}$
  into~$\pi_{g}$, $1 \leqslant g \leqslant h \leqslant n$, split
  the two states $\bpair{\sigma_1}{2^k}$
  and~$\bpair{\sigma_2}{2^k}$. By Lemma~\ref{lem:heights} there are
  $w \in \ofA^\ast$, $\alpha \in \ofA$, and
  $\sigma'_1, \sigma'_2 \in \Bit^k$ such that
    \begin{displaymath}
      \bpair{\sigma_1}{2^k} \pijlsterc{w} \apair{\sigma_1}{w}
      \cpijl{\alpha} \bpair{\sigma_1'}{1}
      \quad \text{and} \quad
      \bpair{\sigma_2}{2^k} \pijlsterc{w} \apair{\sigma_2}{w}
      \cpijl{\alpha} \bpair{\sigma_2'}{1}
    \end{displaymath}
    with $\bpair{\sigma_1'}{1}$ and~$\bpair{\sigma_2'}{1}$ in different
    blocks of~$\pi_{g{-}1}$. Hence, $\sigma'_1$ and~$\sigma'_2$ are in
    different blocks of~$\pi^\ell_{i{-}1}$ while $\sigma_1 \bpijl{a_j}
    \sigma'_1$ and $\sigma_2 \bpijl{a_j} \sigma'_2$ for~$j = \lbl(w
    \mkern1mu \alpha)$, which justifies splitting $\sigma_1$
    and~$\sigma_2$ for $\pi^\ell_i$. We conclude that $\Pi^\ell$ is a
    valid refinement sequence for~$\calB_k$.
	
    We have established that if $\Pi$ is a valid refinement sequence
    for~$\calC_k$, then $\Pi^\ell$ is a valid refinement sequence
    for~$\calB_k$. The sequence~$\Pi^\ell$ is obtained from~$\Pi$ by
    sifting out the blocks of~$\Pi$'s partitions and removing repeated
    partitions. Therefore it holds that
    $\IRC(\Pi) \geqslant \IRC(\Pi^\ell)$. Since the mappings $p_\ell$
    and~$p_{\ell'}$ include pairwise distinct sets of stake states
    for~$\ell \neq \ell'$, $1 \leqslant \ell \leqslant 2^k$, it
    follows that
    \begin{math}
      \IRC(\Pi) \geqslant \sum_{\ell=1}^{2^k} \: \IRC(\Pi^\ell)
      \geqslant 2^k \IRC(\calB_k) \mkern1mu .
    \end{math}
    Taking the minimum over all valid refinement sequences
    for~$\calC_k$ we conclude that
    $ \IRC(\calC_k) \geqslant 2^k \IRC(\calB_k)$ as was to be shown.
\end{proof}

\noindent
With the above technical lemma in place, we are able to strengthen the
$\Omega(n \log n)$ lowerbound of Theorem~\ref{thm:inherent_complexity}
by now taking the number of transitions into account. The improved lowerbound
is $\Omega((m+n)\log n)$, where $m$ is the number of transitions and 
$n$ the number of states.

\begin{thm}
  \label{thm:inflated_inherent_complexity}
  Deciding bisimilarity for (deterministic) LTSs with a partition
  refinement algorithm is $\Omega((m+n)\log n)$, where $n$~is the
  number of states and $m$~is the number of transitions of the LTS\@.
\end{thm}

\begin{proof}
  For the bisplitter~$\calB_k$, we know by
  Theorem~\ref{thm:inherent_complexity} that
  $\IRC(\calB_k) \geqslant 2^{k{-}1} (k{-}1)$. Thus, by
  Lemma~\ref{lemma:inflate}, we obtain
  $\IRC(\calC_k) \geqslant 2^{2k{-}1} (k{-}1)$. In the case
  of~$\calC_k$ we have for $n$ and~$m$ that
  $n = 2^k (2^k + 2^{\lceil \mkern1mu \log (k-1) \rceil}-1)$ and $m = 2n$.
  Hence $n+m \in \Theta( 2^{2k{-}1} )$ and $\log n \in \Theta(k-1)$,
  from which it follows that~$\IRC(\calC_k) \in \Omega((m+n) \log n)$.
\end{proof}

\noindent
Underlying the proof of the lowerbound for deciding bisimilarity for the
family of layered bisplitters~$\calC_k$ is the observation that
each~$\calC_k$ can be seen as $2^k$~stacked instances of the ordinary
bisplitters~$\calB_k$, augmented with tree gadgets to handle
transitions properly. The other essential ingredient for the proof of
Theorem~\ref{thm:inflated_inherent_complexity} is the complexity of
deciding bisimilarity with a partition refinement algorithm on the
$\calB_k$ family. The same reasoning applies when considering
partition refinement algorithms with an oracle for end structures from
Section~\ref{sec:oracle}. Also with an oracle the lowerbound
of~$\Omega((m{+}n) \log n)$ remains.
	
\begin{thm}
  Any partition refinement algorithm with an oracle for end structures
  that decides bisimilarity for (deterministic) LTSs is
  $\Omega((m+n) \log n)$.
\end{thm}

\begin{proof}
  The proof is similar to that of~Lemma~\ref{lemma:oracle_projection}
  and Theorem~\ref{thm:inflated_inherent_complexity}. Consider, for
  some~$k > 2$, the layered bisplitter~$\calC_k$ having initial
  partition~$\pi_0^\calC$\!. The LTS~$\calC_k$ has two end structures,
  viz.\ the set~$S_0 \subseteq S^\calC_k$ containing the states of the
  stake and accompanying tree gadget
  $S_0 = \lc \bpair{\bit{0}^k}{\ell} \mid 1 \leqslant \ell \leqslant
  2^k \rc \cup \lc \apair{\bit{0}^k}{w} \mid w \in \ofA^\ast ,\, |w|
  \leqslant \lceil \log (\frac{k-1}{2}) \rceil\}$ for~$\bit{0}^k$ and
  a similar end structure~$S_1 \subseteq S^\calC_k$
  for~$\bit{10}^{k{-}1}$. The sets $S_0$ and~$S_1$ are minimally
  closed under the transitions of~$\calC_k$. Other states, on the
  stake or tree gadget for a string~$\sigma$, have a path to these
  sets inherited from a path from $\sigma$ to $\bit{0}^k$
  or~$\bit{10}^k$ in~$\calB_k$. The bisimulation classes $S'_0$
  and~$S'_1$, say, with respect to~$S^\calC_k$ rather
  than~$\pi_0^\calC$, consist of $S_0$ and~$S_1$ themselves plus a
  part of the tree gadgets for transitions in~$\calC_k$ leading to
  $S_0$ and~$S_1$, respectively.

  The update of the initial partition~$\pi_0^\calC$ with oracle information,
  which concerns, ignoring the tree gadgets, the common refinement of
  the layers $\lc \bpair{\sigma}{\ell} \mid \sigma \in B_0 \rc$ and
  $\lc \bpair{\sigma}{\ell} \mid \sigma \in B_1 \rc$ on the one hand,
  and 
  the bisimulation classes~$S'_0$ and~$S'_1$ on the other hand, is therefore equal to~$\pi_0^\calC$ on the
  stakes, and generally finer on the tree gadgets.

  Next, every valid refinement sequence
  $\Pi = ( \pi_0', \pi_2, \ldots, \pi_n )$ for the updated LTS
  $\calC'_k = (S, \Act, {\rightarrow}, \pi'_0)$ satisfies
  $\IRC(\Pi) \geqslant \IRC(\calC_{k{-}2})$. Following the lines of
  the proof of Lemma~\ref{lemma:oracle_projection}, we can show that a
  valid refinement sequence~$\Pi$ for $\calC_k$ with updated initial
  partition~$\pi'_0$ induces a valid refinement sequence~$\Pi'$ for
  $\calC_{k-2}$.

  The number of states in $\calC_{k{-}2}$ is $\Theta(n)$ with~$n$ the
  number of states of~$\calC_k$, and the number of transitions
  in~$\calC_{k{-}2}$ is $\Theta(m)$ with~$m$ the number of transitions
  of~$\calC_k$. Therefore,
  $\IRC(\Pi) \geqslant \IRC(\Pi') \geqslant \IRC(\calC_{k{-}2})$, from
  which we derive that any partition refinement algorithm with an oracle
  for end structures involves $\Theta((m{+}n) \log n)$ times moving a
  state for~$\calC_k$, and hence, the algorithm is
  $\Omega((m{+}n) \log n)$.
\end{proof}



\section{An $\Omega(n)$ lowerbound for parallel partitioning algorithms}
\label{sec:parallel}

In this section we pose the question of the effect of the concept of
valid refinements on parallel partition refinement algorithms. We show
an $\Omega(n)$ lowerbound.  This result was already suggested in
\cite[Theorem 3]{Kulakowski2013}, without making explicit which
operations are allowed to calculate the refinement. In particular, for
deterministic LTSs with singleton alphabets, an $\bigO(\log n)$
parallel refinement algorithm~\cite{jaja1994efficient} exists, defying
the argumentation of \cite{Kulakowski2013}. This latter algorithm
clearly is not based on valid refinements.

Parallel bisimulation algorithms are most conveniently studied in the
context of PRAMs (Parallel Random Access Machine)
\cite{StockmeyerV84}, which have an unbounded
number of processors that can all access the available memory. PRAMs
are approximated by GPUs (Graphical Processing Units) that
currently contain thousands of processor cores, but more
interestingly, in combination with the operating system, can run
millions of independent threads simultaneously.

There are a few variants of the PRAM model. The most important
variation is in what happens when multiple processors try to
write to the same address in memory. In the \emph{common} scheme a
write to a particular address takes place if all processors writing to
this address write the same value. Otherwise, the write fails and the
address will contain an arbitrary value. In the \emph{arbitrary}
scheme, one of the processors writing to the address will win,
and writes its value; the writes of other processors to the address
are ignored. In the \emph{priority} scheme, the processor with the lowest
index writes to the address.

A number of algorithms have been proposed to calculate bisimulation on
PRAMs or GPUs \cite{lee_1994_RCPP,rajasekaran1998parallel, wijs_2015,
  martens2021linear}, and there are also parallel algorithms developed
for networks of parallel computers~\cite{blom2005distributed}.  The
algorithms in \cite{lee_1994_RCPP,rajasekaran1998parallel} require
$\bigO(n \log n)$ time on respectively
$\frac{m}{\mkern1mu \log n \mkern1mu}\log\log n$ and
$\frac{m}{n}\log n$ processors.  The algorithm in
\cite{martens2021linear} has the best worst-case
time complexity of~$\bigO(n)$ and uses $\max(n,m)$ processors.  All
these parallel algorithms have in common that they can be classified
as partition refinement algorithms in the sense that they all
calculate a valid sequence of partitions.

Note that parallel refinement algorithms can fundamentally 
outperform sequential algorithms.
In order to understand why parallel algorithms achieve an 
upperbound of ~$\bigO(n)$ vs.\ a lowerbound of 
$\Omega((m+n)\log n)$ for sequential algorithms, we look at the
algorithm in \cite{martens2021linear} in more detail as it has the
best time complexity. For the sake of exposition we assume here that
there is only one action, although the story with multiple actions is
essentially the same.  In the algorithm, first an unstable block is
chosen. All states reaching this block are marked, which is done in
constant time, by one processor per transition. Subsequently, all
marked states in each block separate themselves from the other states
in constant time using one processor per state, by employing an
intricate trick where each block is characterised by a unique `leader'
state in each block. Here it is essential that the PRAM model
uses the arbitrary or priority writing scheme. The algorithm does not
work in the common scheme.  In \cite{martens2021linear} it is shown
that at most $3n$ of these constant-time splitting
steps need to be performed, leading to a time complexity of
$\bigO(n)$.

An important observation is that parallel algorithms allow to split
blocks in constant time, whereas for sequential algorithms we defined
the refinement costs as the minimal number of states that had to be
moved from one block to a new block.
We assume, similarly to the sequential setting,
that a new refinement can only be calculated on the
basis of a previous refinement.  This makes it natural to define the
notion of parallel refinement costs as the minimal
conceivable length of a valid refinement sequence and take this as the
minimal time required to calculate bisimulation using partitioning in
a parallel setting.

For an LTS $L$ and a sequence $\Pi = (\pi_0, \dots, \pi_n)$, the
\textit{parallel refinement cost} is the number of refinements in the
sequence, $\PIRC(\Pi) = n$. For an LTS $L$ we define
\begin{displaymath}
  \PIRC(L) = \min \lc \PIRC{(\Pi)} \mid
  \text{$\Pi$ is a valid refinement sequence for $L$} \rc
  \mkern1mu .
\end{displaymath}
Observe that parallel refinement costs allow for
extremely fast partitioning of transition systems. Below we
show an example with $2^{k}+k$ states with a refinement cost
of~$1$. The states are given by $b_0, \ldots, b_{k{-}1}$,
$a_0,\ldots,a_{2^k-1}$. There is a transition from $a_i$ to $b_j$ if
the $j$-th bit in the binary representation of~$i$ is~$1$. The initial
partition $\pi_\textit{init}$ groups all states~$a_i$ in one
partition, and puts each state~$b_j$ in a partition of its own. So,
$\pi_\textit{init}$ contains $k{+}1$~blocks.  In
Figure~\ref{fastparallelpartition} this transition system is depicted
for~$k=3$.

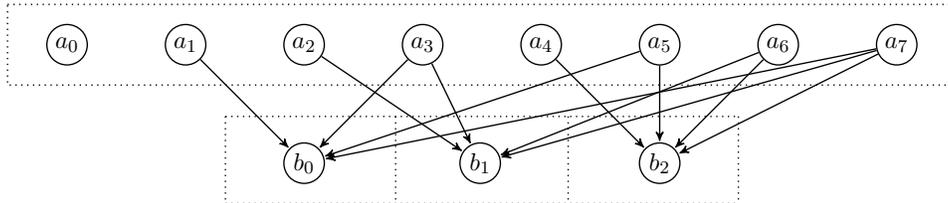
\begin{figure}[htb]
  \scalebox{0.85}{%
\begin{tikzpicture}[scale=1,
every node/.style={scale=1},
node distance = 1.2cm, initial text=]

\node[state] (0) {$a_0$};
\node[state, right= of 0] (1) {$a_1$};
\node[state, right= of 1] (2) {$a_2$};
\node[state, right= of 2] (3) {$a_3$};
\node[state, right= of 3] (4) {$a_4$};
\node[state, right= of 4] (5) {$a_5$};
\node[state, right= of 5] (6) {$a_6$};
\node[state, right= of 6] (7) {$a_7$};
\node[state, below= of 2] (b0) {$b_0$};
\node[state, right=2.1 of b0] (b1) {$b_1$};
\node[state, below= of 5] (b2) {$b_2$};

\path[->]
(1) edge (b0)
(2) edge (b1)
(3) edge (b0)
(3) edge (b1)
(4) edge (b2)
(5) edge (b0)
(5) edge (b2)
(6) edge (b1)
(6) edge (b2)
(7) edge (b0)
(7) edge (b1)
(7) edge (b2);

\coordinate [above left=4mm and 7mm of 0] (linehighA);
\coordinate [below left=4mm and 7mm of 0]  (linelowA);
\coordinate [above right=4mm and 7mm of 7] (linehighB);
\coordinate [below right=4mm and 7mm of 7] (linelowB);

\draw [dotted, -] (linehighA) -- (linehighB) -- (linelowB) --
(linelowA) -- cycle ;

\coordinate [above left=5mm and 10mm of b0] (linehigh1);
\coordinate [below left=4mm and 10mm of b0]  (linelow1);
\coordinate [above right=5mm and 12mm of b0] (linehigh2);
\coordinate [below right=4mm and 12mm of b0] (linelow2);
\coordinate [above left=5mm and 12mm of b2] (linehigh3);
\coordinate [below left=4mm and 12mm of b2] (linelow3);
\coordinate [above right=5mm and 10mm of b2] (linehigh4); 
\coordinate [below right=4mm and 10mm of b2] (linelow4);

\draw [dotted, -] (linehigh1) -- (linehigh4);
\draw [dotted, -] (linelow1) -- (linelow4);
\draw [dotted, -] (linehigh1) edge (linelow1);
\draw [dotted, -] (linehigh2) edge (linelow2);
\draw [dotted, -] (linehigh3) edge (linelow3);
\draw [dotted, -] (linehigh4) edge (linelow4);

\end{tikzpicture}
  } 
  \caption{A transition system with a parallel refinement cost of $1$}
  \label{fastparallelpartition}
\end{figure}

\noindent
The shortest valid refinement sequence is
$(\pi_\textit{init},\pi_\textit{final})$, where in
$\pi_\textit{final}$ each state is in a separate block. This
refinement is valid, because in $\pi_\textit{init}$ there is enough
information to separate each state from any other, as can easily be
checked against the definition. As this refinement sequence has length
$1$, the parallel refinement cost of this transition system is~$1$,
indicating that it is conceivable to make a bisimulation partitioning
algorithm doing this refinement in constant time. Note that existing
parallel algorithms do not achieve this performance. For instance, the
algorithm in \cite{martens2021linear} requires linear time as it
checks stability for each new block sequentially.

Although parallel partitioning can be fast, we show, using the notion
of parallel refinement costs, that calculating bisimulation
in parallel requires time $\Omega(n)$. For this purpose, we construct a family of
LTSs $\calD_n$ for which the length of any valid refinement sequence
grows linearly with the number of states.

\begin{defi} 
  \label{def:seqsplitter}
  For $n > 2$, the sequential splitter
  $\calD_n = (S, \{a\}, {\rightarrow}, \pi_0)$ is defined as the LTS
  that has the set $S = \lc 1, \dots, n \rc$ as its set of states, the
  relation
  \begin{displaymath}
    {\pijl{}} = \lc (i,i{+}1) \mid 1 \leqslant i < n\} \cup \{ (n,n) \}
  \end{displaymath}
  as the transition relation, and the set
  $\pi_0 = \lc \{ 1, \ldots, n{-}1 \rc, \lc n \} \rc$ as its initial
  partition.
\end{defi}

\noindent
For every $n > 2$, the deterministic LTS
$\calD_n = (S,\Sigma, \rightarrow, \pi_0)$ has $n$~states and
$n$~transitions. For~$n=8$ the transition system is
depicted in Figure~\ref{fig:sequential-example}. The set of states
$S = \{1,\dots, n\}$ form a chain in which every state
$i\in \{1, \dots, n{-}1\}$ has an outgoing transition to the next
state $(i, i{+}1) \in {\rightarrow}$. The state at the end of the
chain~$n \in S$ has a self loop $(n,n)\in {\rightarrow}$. In the
initial partition there are two blocks, one block containing~$n \in S$
and the other block containing all other states. 

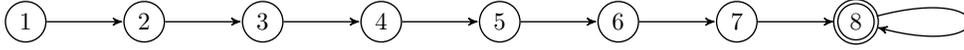
\begin{figure}[bht]
  \begin{center}

\scalebox{0.85}{%
\begin{tikzpicture}[scale=1,
every node/.style={scale=1},
node distance = 1.2cm, initial text=]

\node[state] (1) {$1$};
\node[state, right= of 1] (2) {$2$};
\node[state, right= of 2] (3) {$3$};
\node[state, right= of 3] (4) {$4$};
\node[state, right= of 4] (5) {$5$};
\node[state, right= of 5] (6) {$6$};
\node[state, right= of 6] (7) {$7$};
\node[state, accepting , right= of 7] (8) {$8$};

\path[->]
(1) edge (2)
(2) edge (3)
(3) edge (4)
(4) edge (5)
(5) edge (6)
(6) edge (7)
(7) edge (8)

(8) edge [loop right, min distance=2cm] (8);

\end{tikzpicture}
} 
  \end{center}
  \caption{Sequential splitter $\calD_8$ with initial partition
    $\{\{1,2,3,4,5,6,7\}, \{8\}\}$ 
    \label{fig:sequential-example}} 
\end{figure}

\noindent
The following lemma states that for every sequential splitter
$\calD_n$ there is a unique valid refinement sequence.

\begin{lem}\label{lemma:longsplitter}
  For every $n > 2$, $\calD_n$ has a \emph{unique} valid refinement
  sequence $\Pi_{n}$ that consists of $n{-}1$~partitions.
\end{lem}

\begin{proof}
  For $n > 2$ the sequential splitter $\calD_n$ has a valid partition
  refinement sequence which is given by
  $\Pi_{\calD_n} = ( \pi_1, \dots, \pi_{n{-}1} )$, where
  $\pi_i = \lc \{1,\dots, n{-}i\},
  \{n{-}i{+}1\},\{n{-}i{+}2\},\dots,\{n\} \rc$. This is proven by
  induction on the index~$i$ of the partition~$\pi_i$.

  Next, we must show that this refinement sequence is unique. So, in
  order to obtain a contradiction, consider some $\pi_i$ and assume
  some valid refinement $\pi'_{i+1}$ different from $\pi_{i+1}$
  exists. This means that there must either be two states
  $j,j' \in \{1,\ldots,n{-}i{-}1\}$ that are in different blocks in
  $\pi'_{i+1}$, or state~$n{-}i$ must be in the same block as some
  state~$j <n{-}i$. In the first case states $j$ and~$j'$ are in the same block
  in~$\pi_i$ and have exactly the same transitions to the same block
  in~$\pi_i$. Hence, the states $j$ and $j'$ are in the same
  block in $\pi_{i{+}1}$.
  In the second case, all states $j'<n{-}i$ must be in the same block as
  $j$, using exactly the argument of the first case. The state $n{-}i$
  is also in that same block. But then $\pi'_{i+1}$ is not a strict
  refinement of $\pi_i$, making it invalid.
    
  Thus, for each~$i$, $1 \leqslant i \leqslant n{-}1$, no other valid
  refinement than $\pi_{i+1}$ of $\pi_i$ exists, making~$\Pi_{n}$
  the only valid refinement sequence for~$\calD_n$.
\end{proof}	

\noindent
The observation in Lemma~\ref{lemma:longsplitter} leads to the
following theorem on the time complexity for parallel partition
refinement algorithms.

\begin{thm}
  Any parallel partition refinement algorithm that decides
  bisimilarity for an LTS with $n$~states has
  time complexity~$\Omega(n)$.
\end{thm}

\begin{proof}
  For every $n \geqslant 2$ the LTS $\calD_n$ has $n$~states and
  $n$~transitions. Any algorithm that is a parallel partition
  refinement algorithm has time complexity that is at least the length
  of the shortest refinement sequence.  According to
  Lemma~\ref{lemma:longsplitter} $\calD_n$ has a unique refinement
  sequence which witnesses $\PIRC(\calD_n) = n{-}1 \in \Omega(n)$.
\end{proof}

\noindent
The theorem puts a bound on the fastest possible
parallel, partition based algorithms for bisimulation.  But it must be
observed that other techniques than partition refinement can produce
faster algorithms, although it may be in more restricted settings.
Concretely, the algorithm~\cite{jaja1994efficient} that is based on the 
sequential Roberts' algorithm from Section~\ref{sec:roberts}, runs in
$\bigo{\log n}$ parallel time for deterministic transition systems
with $n$~states and only one transition label.

In order to obtain an idea of how these fast parallel algorithms work,
we illustrate one of the major techniques to determine a bisimulation
refinement of~${\calD_n}$ in time $\bigo{\log n}$.  In~${\calD_n}$ the
initial partition has a small block $B_\textit{small}=\{n\}$ and a
large block $B_\textit{large} = \{ 1, \ldots, n{-}1 \}$.  States
in~$B_\textit{large}$ with a different number of steps to
$B_\textit{small}$ cannot be bisimilar and can be split.  Note that
splitting on the basis of this distance is not a valid refinement in
the sense of Definition~\ref{def:bis_valid}.

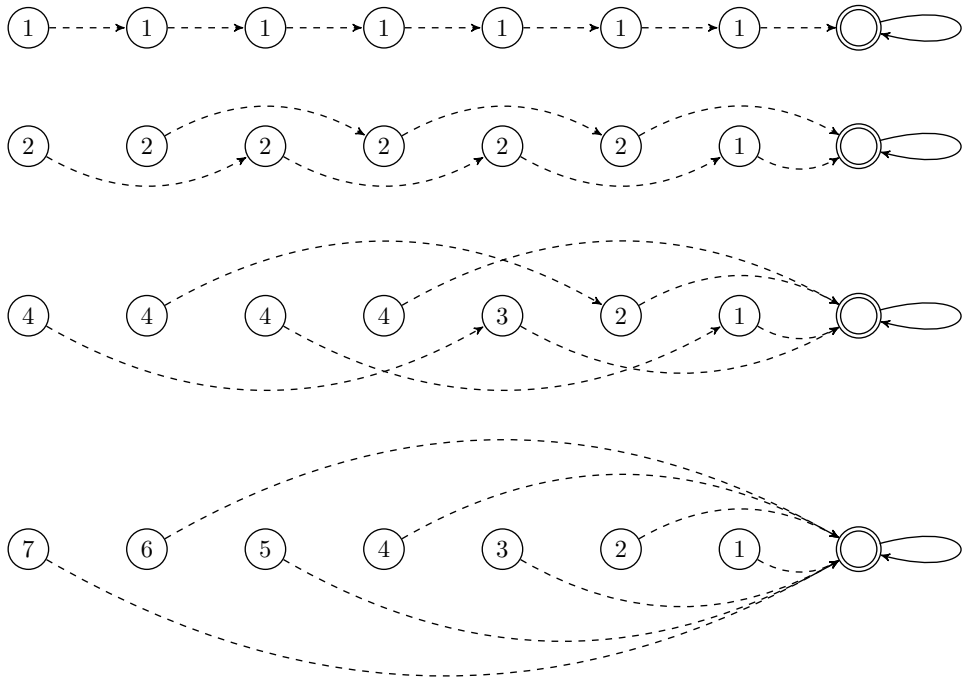
\begin{figure}[htb]
  \scalebox{0.85}{%
\begin{tikzpicture}[scale=1,
every node/.style={scale=1},
node distance = 1.2cm, initial text=]

\node[state] (11) {$1$};
\node[state, right= of 11] (21) {$1$};
\node[state, right= of 21] (31) {$1$};
\node[state, right= of 31] (41) {$1$};
\node[state, right= of 41] (51) {$1$};
\node[state, right= of 51] (61) {$1$};
\node[state, right= of 61] (71) {$1$};
\node[state,accepting, right= of 71] (81) {$ $};

\path[->, dashed]
(11) edge (21)
(21) edge (31)
(31) edge (41)
(41) edge (51)
(51) edge (61)
(61) edge (71)
(71) edge (81)

(81) edge[loop right, min distance=1.75cm, solid] (81);

\node[state, below= of 11] (12) {$2$};
\node[state, right= of 12] (22) {$2$};
\node[state, right= of 22] (32) {$2$};
\node[state, right= of 32] (42) {$2$};
\node[state, right= of 42] (52) {$2$};
\node[state, right= of 52] (62) {$2$};
\node[state, right= of 62] (72) {$1$};
\node[state,accepting, right= of 72] (82) {$ $};

\path[->, dashed]
(12) edge[bend right] (32)
(22) edge[bend left]  (42)
(32) edge[bend right] (52)
(42) edge[bend left]  (62)
(52) edge[bend right] (72)
(62) edge[bend left]  (82)
(72) edge[bend right] (82)

(82) edge[loop right, min distance=1.75cm, solid] (82);

\node[state, below=2 of 12] (13) {$4$};
\node[state, right= of 13] (23) {$4$};
\node[state, right= of 23] (33) {$4$};
\node[state, right= of 33] (43) {$4$};
\node[state, right= of 43] (53) {$3$};
\node[state, right= of 53] (63) {$2$};
\node[state, right= of 63] (73) {$1$};
\node[state,accepting, right= of 73] (83) {$ $};

\path[->, dashed]
(13) edge[bend right] (53)
(23) edge[bend left]  (63)
(33) edge[bend right] (73)
(43) edge[bend left]  (83)
(53) edge[bend right] (83)
(63) edge[bend left]  (83)
(73) edge[bend right] (83)

(83) edge[loop right, min distance=1.75cm, solid] (83);

\node[state, below=3 of 13] (14) {$7$};
\node[state, right= of 14] (24) {$6$};
\node[state, right= of 24] (34) {$5$};
\node[state, right= of 34] (44) {$4$};
\node[state, right= of 44] (54) {$3$};
\node[state, right= of 54] (64) {$2$};
\node[state, right= of 64] (74) {$1$};
\node[state,accepting, right= of 74] (84) {$ $};

\path[->, dashed]
(14) edge[bend right] (84)
(24) edge[bend left]  (84)
(34) edge[bend right] (84)
(44) edge[bend left]  (84)
(54) edge[bend right] (84)
(64) edge[bend left]  (84)
(74) edge[bend right] (84)

(84) edge[loop right, min distance=1.75cm, solid] (84);

\end{tikzpicture}
  } 
\caption{Calculating the distance to the rightmost state in
  $\bigO(\log n)$ time} 
\label{fig:counting}
\end{figure}

\noindent
Determining the distance of the states in $B_{\textit{large}}$ to
$B_\textit{small}$ can be done in $\bigO(\log n)$ time in parallel,
cf.~\cite{Hillis86}. The basic idea is explained in
Figure~\ref{fig:counting}.  Each state in the block at the left gets
weight~$1$. We desire to sum up in each state the weights of all
states to its right. We do this by adding up the weight of the right
neighbour and adapting the outgoing transition to point to the state
to which the right neighbour is pointing. We use dashed transitions to
stress that we are now using the transitions for another purpose. In
each round~$k$ a state contains the sum of all $2^{k}$
states to its right. So, after $\log_2 k$ rounds it contains the sum
of all $k$~states to its right, which is the distance
to~$B_\textit{small}$.
 
An interesting open question is whether the notion of a valid
refinement sequence can be adapted, such that the $\Omega(n)$
lowerbound would still apply when techniques such as parallel counting
formulated above would be incorporated in the bisimulation algorithm.


\section{Colour refinement}
\label{sec:color-refinement}

As for establishing bisimilarity on labelled transition systems,
partition refinement is the standard approach for algorithms that do
colour refinement.
Colour refinement, also known as naive vertex classification or 
1-dimensional Weisfeiler-Lehman test,
is frequently applied, among others in the setting of deciding graph
isomorphism~\cite{GN21:cacm}.
Given a graph where each node has been assigned an initial colour,
colour refinement asks to find a refining colouring with the least
number of colours possible such that two nodes of the same colour have,
for all colours, the same number of neighbours of the latter colour. Also
for colour refinement, algorithms typically search for the coarsest
stable colouring.
See~\cite{GKMS21:mit} for
an overview.

Colour refinement is known to be $\Omega((m{+}n) \mkern1mu \log n)$
with $n$ the number of nodes and $m$~the number of edges.
Early algorithms of complexity $\bigo{(m{+}n) \mkern1mu \log n}$
include~\cite{CC82:tcs,paige1987three}. The lowerbound for partition
refinement algorithms has been established
in~\cite{Ber15:phd,berkholz2017tight}.  However, for colour refinement
the costs are measured in terms of inspected edges and hence are different
from the costs for computing bisimilarity by partition refinement.
The cost function underlying the complexity of colour refinement is
defined by
\begin{displaymath}
  \mathit{cost} (R,S) =
  | \mkern1mu \lc (u,v) \in E \mid u \in R ,\, v \in S \rc \mkern1mu |
  \mkern1mu ,
\end{displaymath}
i.e.\ the number of transitions between the blocks $R$ and~$S$ (sets
including all nodes of chosen colours), where block~$R$ is recoloured
in view of the colours in block~$S$.

The paper~\cite{berkholz2017tight} provides a detailed implementation
of an efficient algorithm with time complexity $\bigo{(m{+}n) \log n}$, 
that given a graph $G = (V,E)$ and colouring~$\alpha$
finds the coarsest stable colouring refining~$\alpha$ and compares the
algorithm to other proposals in the literature.
In this setting, a colouring~$\gamma : V \to \ofN$ of~$G$ is stable iff
$|\calN(u) \cap \gamma^{-1}(c)| = |\calN(v) \cap \gamma^{-1}(c)|$ for all
nodes $u,v \in V$ and all colours~$c \in \ofN$. Here, $\calN(u)$
and~$\calN(v)$ denote the sets of neighbours of $u$ and~$v$ in~$G$,
respectively.
In order to establish a lowerbound, Berkholz et al.\ define a family
of graphs~$\calG_k$
for which the costs of computing the coarsest stable colouring is
$\Omega({(m{+}n) \log n})$ starting from the unit colouring assigning
to all nodes the same colour. The paper also discusses the connection
of colour refinement with equivalence in 2-variable logic and with
finding bisimilarity on Kripke structures. Regarding the latter the
focus is on Kripke structures rather than labelled transition systems,
as in the present paper.  In order to transfer the lowerbound result
for colour refinement to a lowerbound result for bisimilarity by
partition refinement for Kripke structures, the stability requirement
for colouring mentioned above is adapted, viz.\ to
$\calNplus(u) \cap S = \emptyset \iff \calNplus(v) \cap S = \emptyset$ 
for all blocks~$R,S$ and nodes $u,v$ in~$R$, with $\calNplus(u)$ 
and~$\calNplus(v)$, i.e.\ the directly reachable states for $u$ 
and~$v$, respectively.

Central to the resulting family of Kripke structures~$\calS_k$
in~\cite{berkholz2017tight} are complete bipartite graphs $K_{k,k}$
with~$k^2$ transitions, one for each bitstring in~$\Bit^k$, which are
dense with respect to transitions. This is because the refinement
costs incurred for colour refinement are based on counting edges. The
family of LTSs~$\calC_k$ presented in
Section~\ref{sec:two-action-action-set} has as their main components
thin stakes that are $2^{k}$~states high, one for every
bitstring in~$\Bit^k$ because for partition refinement for
bisimulation the number of (moved) states is relevant.

Although the families of graphs for the lowerbounds of bisimulation
and graph colouring are definitely related, they are very different if
it comes to the approach, in particular regarding the measurement of
complexity. It is unclear how to transform the $\calS_k$-family into a
family of labelled transition systems such that partition refinement
takes $\Omega({(m{+}n) \log n})$ transfers of states to a newly created
block, i.e.\ in terms of the refinement cost~$\mathit{rc}$ of
Section~\ref{sec:prelims}. Still the Kripke structures~$\calS_k$
of~\cite{berkholz2017tight} can be interpreted as non-deterministic
labelled transition systems with a single action label and an initial
partition based on the assignment of atomic proposition. Similarly,
it is not obvious how to transform the $\calC_k$-family into a family
of undirected or directed graphs such that colour refinement requires
inspection of $\Omega((m{+}n) \log n)$ edges.


\section{Conclusion}
\label{sec:conclusion}

We have shown that, even when restricted to deterministic LTSs, it is
not possible to construct a sequential algorithm based on partition
refinement that is more efficient than $\Omega((m+n)\log n)$.
The bound obtained is preserved even when the algorithm is extended
with an oracle that can determine for specific states in constant time 
whether they are bisimilar or not. The oracle proof technique enabled
us to show that the algorithmic ideas underlying Roberts' algorithm
\cite{paige1985linear} for the one-letter alphabet case cannot be used
to come up with a fundamentally faster enhanced partition refinement
algorithm for bisimulation.

Of course, this is not addressing a generic lower bound
to decide bisimilarity on LTSs, nor proving the conjecture that the
Paige-Tarjan algorithm is optimal for deciding bisimilarity.
It is conceivable that a more efficient algorithm for bisimilarity 
exists that is not based on partitioning.
However, as it stands, no techniques are known to prove
such a generic algorithmic lowerbound,
and all techniques that do exist make assumptions on allowed
operations, such as the well-known lowerbound on sorting.

But by relaxing the notion of a valid partition sequence, and maybe
introducing alternatives for oracles, it may very well be possible
that the lower bound is extended to a wider range of algorithmic
techniques to determine bisimulation, making it very unlikely that
sequential algorithms for bisimulation with a time-complexity better
than $\bigO((m+n)\log n)$ exist. Note that the current lowerbound
already applies to all known efficient algorithms for bisimulation.

For the parallel setting, we showed that deciding bisimilarity by
partitioning is~$\Omega(n)$.  In this case a similar situation
occurs. For LTSs with one action label it is possible to calculate
bisimulation in logarithmic time, cf.~\cite{jaja1994efficient}. An
interesting, but as yet open question is whether the techniques used
\cite{jaja1994efficient} can fundamentally improve the efficiency of
determining bisimulation in parallel, or, as we believe, the
lowerbound result can be strengthened along the lines of the
sequential case to show that the techniques of
\cite{jaja1994efficient} are insufficient to obtain a sub-linear
parallel time complexity to determine bisimulation for
labelled transition systems with at least two action labels.

\bibliographystyle{alphaurl}
\bibliography{main} 

\newcommand{\etalchar}[1]{$^{#1}$}
\begin{thebibliography}{WDMS20}

\bibitem[BBG17]{berkholz2017tight}
C.~Berkholz, P.~Bonsma, and M.~Grohe.
\newblock Tight lower and upper bounds for the complexity of canonical colour
  refinement.
\newblock {\em Theory of Computing Systems}, 60(4):581--614, 2017.
\newblock \href {https://doi.org/10.1007/s00224-016-9686-0}
  {\path{doi:10.1007/s00224-016-9686-0}}.

\bibitem[BC04]{berstel2004complexity}
J.~Berstel and O.~Carton.
\newblock On the complexity of {Hopcroft}’s state minimization algorithm.
\newblock In M.~Domaratzki et~al., editor, {\em Proc.\ CIAA 2004}, volume 3317
  of {\em Lecture Notes in Computer Science}, pages 35--44. Springer, 2004.
\newblock \href {https://doi.org/10.1007/978-3-540-30500-2_4}
  {\path{doi:10.1007/978-3-540-30500-2_4}}.

\bibitem[Ber15]{Ber15:phd}
C.~Berkholz.
\newblock {\em Lower Bounds for Heuristic Algorithms}.
\newblock PhD thesis, RWTH Aachen, 2015.

\bibitem[BO05]{blom2005distributed}
S.~Blom and S.~Orzan.
\newblock A distributed algorithm for strong bisimulation reduction of state
  spaces.
\newblock {\em Software Technology for Technology Transfer}, 7(1):74--86, 2005.
\newblock \href {https://doi.org/10.1007/s10009-004-0159-4}
  {\path{doi:10.1007/s10009-004-0159-4}}.

\bibitem[Buc99]{buchholz199lumping}
P.~Buchholz.
\newblock Exact performance equivalence: An equivalence relation for stochastic
  automata.
\newblock {\em Theoretical Computer Science}, 215:263--287, 1999.
\newblock \href {https://doi.org/10.1016/S0304-3975(98)00169-8}
  {\path{doi:10.1016/S0304-3975(98)00169-8}}.

\bibitem[CC82]{CC82:tcs}
A.~Cardon and M.~Crochemore.
\newblock Partioning a graph in $\mathit{O}(|{A}| \log_2 |{V}|)$.
\newblock {\em Theoretical Computer Science}, 19(1):85--98, 1982.
\newblock \href {https://doi.org/10.1016/0304-3975(82)90016-0}
  {\path{doi:10.1016/0304-3975(82)90016-0}}.

\bibitem[CRS08]{castiglione2008hopcroft}
G.~Castiglione, A.~Restivo, and M.~Sciortino.
\newblock Hopcroft’s algorithm and cyclic automata.
\newblock In C.~Mart{\'{\i}}n-Vide et~al., editor, {\em Proc.\ LATA 2008},
  volume 5196 of {\em Lecture Notes in Computer Science}, pages 172--183.
  Springer, 2008.
\newblock \href {https://doi.org/10.1007/978-3-540-88282-4_17}
  {\path{doi:10.1007/978-3-540-88282-4_17}}.

\bibitem[DPP04]{dovier04efficient}
A.~Dovier, C.~Piazza, and A.~Policriti.
\newblock An efficient algorithm for computing bisimulation equivalence.
\newblock {\em Theoretical Computer Science}, 311:221--256, 2004.
\newblock \href {https://doi.org/10.1016/S0304-3975(03)00361-X}
  {\path{doi:10.1016/S0304-3975(03)00361-X}}.

\bibitem[GKMS21]{GKMS21:mit}
M.~Grohe, K.~Kersting, M.~Mladenov, and P.~Schweitzer.
\newblock Color refinement and its applications.
\newblock In G.~Van~den Broek, K.~Kersting, and Natarajan S., editors, {\em An
  Introduction to Lifted Probabilistic Inference}, chapter~15. The MIT Press,
  2021.
\newblock \href {https://doi.org/10.7551/mitpress/10548.003.0023}
  {\path{doi:10.7551/mitpress/10548.003.0023}}.

\bibitem[GMV21]{GMV21:concur}
J.F. Groote, J.~Martens, and E.P.~de Vink.
\newblock Bisimulation by partitioning is ${\Omega}((m+n)\log n)$.
\newblock In S.~Haddad and D.~Varacca, editors, {\em Proc.\ CONCUR 2021},
  volume 203 of {\em LIPIcs}, pages 31:1--31:16. Schloss
  Dagstuhl--Leibniz-Zentrum f{\"u}r Informatik, 2021.
\newblock \href {https://doi.org/10.4230/LIPIcs.CONCUR.2021.31}
  {\path{doi:10.4230/LIPIcs.CONCUR.2021.31}}.

\bibitem[GN21]{GN21:cacm}
M.~Grohe and D.~Neuen.
\newblock Isomorphism, canonization, and definability for graphs of bounded
  rank width.
\newblock {\em Communications of the ACM}, 64(5):98--105, 2021.
\newblock \href {https://doi.org/10.1145/3453943} {\path{doi:10.1145/3453943}}.

\bibitem[GVV18]{groote2018prob}
J.F. Groote, H.J.~Rivera Verduzco, and E.P.~de Vink.
\newblock An efficient algorithm to determine probabilistic bisimulation.
\newblock {\em Algorithms}, 11(9):131,1--22, 2018.
\newblock \href {https://doi.org/10.3390/a11090131}
  {\path{doi:10.3390/a11090131}}.

\bibitem[HJ86]{Hillis86}
W.D. Hillis and G.L.~Steele Jr.
\newblock Data parallel algorithms.
\newblock {\em Communications of the {ACM}}, 29(12):1170--1183, 1986.
\newblock \href {https://doi.org/10.1145/7902.7903}
  {\path{doi:10.1145/7902.7903}}.

\bibitem[Hop71]{hopcroft1971DFAmin}
J.~Hopcroft.
\newblock An $n \log n$ algorithm for minimizing states in a finite automaton.
\newblock In Z.~Kohavi and A.~Paz, editors, {\em Theory of Machines and
  Computations}, pages 189--196. Academic Press, 1971.
\newblock \href {https://doi.org/10.1016/b978-0-12-417750-5.50022-1}
  {\path{doi:10.1016/b978-0-12-417750-5.50022-1}}.

\bibitem[JGKW20]{jansen_2020}
D.N. Jansen, J.F. Groote, J.J.A. Keiren, and A.~Wijs.
\newblock An $\textrm{O}(m \log n)$ algorithm for branching bisimilarity on
  labelled transition systems.
\newblock In A.~Biere and D.~Parker, editors, {\em Proc.\ TACAS}, volume 12079
  of {\em Lecture Notes in Computer Science}, pages 3--20. Springer, 2020.
\newblock \href {https://doi.org/10.1007/978-3-030-45237-7_1}
  {\path{doi:10.1007/978-3-030-45237-7_1}}.

\bibitem[JR94]{jaja1994efficient}
J.~J{\'a}j{\'a} and Kwan~Woo Ryu.
\newblock An efficient parallel algorithm for the single function coarsest
  partition problem.
\newblock {\em Theoretical Computer Science}, 129(2):293--307, 1994.
\newblock \href {https://doi.org/10.1016/0304-3975(94)90030-2}
  {\path{doi:10.1016/0304-3975(94)90030-2}}.

\bibitem[KMP77]{KnuthMP77}
D.E. Knuth, J.H. {Morris}, and V.R. Pratt.
\newblock Fast pattern matching in strings.
\newblock {\em {SIAM} Journal on Computing}, 6(2):323--350, 1977.
\newblock \href {https://doi.org/10.1137/0206024} {\path{doi:10.1137/0206024}}.

\bibitem[KS90]{kanellakis1983}
P.C. Kanellakis and S.A. Smolka.
\newblock {CCS} expressions, finite state processes, and three problems of
  equivalence.
\newblock {\em Information and Computation}, 86(1):43--68, 1990.
\newblock \href {https://doi.org/10.1016/0890-5401(90)90025-D}
  {\path{doi:10.1016/0890-5401(90)90025-D}}.

\bibitem[Kul13]{Kulakowski2013}
K.~Kulakowski.
\newblock Concurrent bisimulation algorithm.
\newblock {\em ArXiv, CoRR}, abs/1311.7635, 2013.

\bibitem[LR94]{lee_1994_RCPP}
I.~Lee and S.~Rajasekaran.
\newblock A parallel algorithm for relational coarsest partition problems and
  its implementation.
\newblock In D.L. Dill, editor, {\em Computer Aided Verification}, volume 818
  of {\em Lecture Notes in Computer Science}, pages 404--414. Springer, 1994.
\newblock \href {https://doi.org/10.1007/3-540-58179-0_71}
  {\path{doi:10.1007/3-540-58179-0_71}}.

\bibitem[MGH{\etalchar{+}}21]{martens2021linear}
J.J.M. Martens, J.F. Groote, L.B. van~den Haak, H.P. Hijma, and A.J. Wijs.
\newblock A linear parallel algorithm to compute bisimulation and relational
  coarsest partitions.
\newblock In Gwen Sala{\"u}n and Anton Wijs, editors, {\em Proc.\ FACS}, volume
  13077 of {\em Lecture Notes in Computer Science}, pages 115--133. Springer,
  2021.
\newblock \href {https://doi.org/10.1007/978-3-030-90636-8_7}
  {\path{doi:10.1007/978-3-030-90636-8_7}}.

\bibitem[Mil80]{milner1980calculus}
R.~Milner.
\newblock {\em A Calculus of Communicating Systems}, volume~92 of {\em Lecture
  Notes in Computer Science}.
\newblock Springer, 1980.
\newblock \href {https://doi.org/10.1007/3-540-10235-3}
  {\path{doi:10.1007/3-540-10235-3}}.

\bibitem[Par81]{park1981gi}
D.M.R. Park.
\newblock Concurrency and automata on infinite sequences.
\newblock In P.~Deussen, editor, {\em Proc.\ 5th~GI-Conference on Theoretical
  Computer Science}, volume 104 of {\em Lecture Notes in Computer Science},
  pages 167--183. Springer, 1981.
\newblock \href {https://doi.org/10.1007/BFb0017309}
  {\path{doi:10.1007/BFb0017309}}.

\bibitem[PT87]{paige1987three}
R.~Paige and R.E. Tarjan.
\newblock Three partition refinement algorithms.
\newblock {\em SIAM Journal on Computing}, 16(6):973--989, 1987.
\newblock \href {https://doi.org/10.1137/0216062} {\path{doi:10.1137/0216062}}.

\bibitem[PTB85]{paige1985linear}
R.~Paige, R.E. Tarjan, and R.~Bonic.
\newblock A linear time solution to the single function coarsest partition
  problem.
\newblock {\em Theoretical Computer Science}, 40:67--84, 1985.
\newblock \href {https://doi.org/10.1016/0304-3975(85)90159-8}
  {\path{doi:10.1016/0304-3975(85)90159-8}}.

\bibitem[RL98]{rajasekaran1998parallel}
S.~Rajasekaran and I.~Lee.
\newblock Parallel algorithms for relational coarsest partition problems.
\newblock {\em IEEE Transactions on Parallel and Distributed Systems},
  9(7):687--699, 1998.
\newblock \href {https://doi.org/10.1109/71.707548}
  {\path{doi:10.1109/71.707548}}.

\bibitem[SV84]{StockmeyerV84}
L.J. Stockmeyer and U.~Vishkin.
\newblock Simulation of parallel random access machines by circuits.
\newblock {\em {SIAM} Journal of Computing}, 13(2):409--422, 1984.
\newblock \href {https://doi.org/10.1137/0213027} {\path{doi:10.1137/0213027}}.

\bibitem[WDMS20]{wissman2020coalgebra}
T.~Wi{\ss}mann, U.~Dorsch, S.~Milius, and L.~Schr{\"{o}}der.
\newblock Efficient and modular coalgebraic partition refinement.
\newblock {\em Logical Methods Computer Science}, 16(1), 2020.
\newblock \href {https://doi.org/10.23638/LMCS-16(1:8)2020}
  {\path{doi:10.23638/LMCS-16(1:8)2020}}.

\bibitem[Wij15]{wijs_2015}
A.J. Wijs.
\newblock {GPU} accelerated strong and branching bisimilarity checking.
\newblock In C.~Baier and C.~Tinelli, editors, {\em Proc.\ TACAS}, volume 9035
  of {\em Lecture Notes in Computer Science}, pages 368--383. Springer, 2015.
\newblock \href {https://doi.org/10.1007/978-3-662-46681-0_29}
  {\path{doi:10.1007/978-3-662-46681-0_29}}.

\end{thebibliography}
\end{document}